\documentclass[journal]{IEEEtran}
\usepackage{amsmath}
\usepackage{tikz}
\usepackage{amssymb}
\usepackage{amsfonts}
\usepackage{algorithm}
\usepackage{algorithmicx}
\usepackage{algpseudocode}
\usepackage{dsfont}
\usepackage{float}
\usepackage{cite}
\usepackage{graphicx}
\usepackage{epsfig}
\usepackage{subfigure}
\usepackage{psfrag}
\usepackage{xcolor}
\usepackage{url}
\usepackage[colorlinks,linkcolor=black,urlcolor=black,anchorcolor=black,citecolor=black,hyperfootnotes=true]{hyperref}
\usepackage{bm}
\usepackage{amsthm}
\theoremstyle{plain}

\newtheorem{theorem}{Theorem}

\theoremstyle{definition}

\def\BibTeX{{\rm B\kern-.05em{\sc i\kern-.025em b}\kern-.08em
    T\kern-.1667em\lower.7ex\hbox{E}\kern-.125emX}}
\usepackage{balance}
\usepackage{lipsum}
\newcommand\blfootnote[1]{%
  \begingroup
  \renewcommand\thefootnote{}\footnote{#1}%
  \addtocounter{footnote}{-1}%
  \endgroup
}

\begin{document}

\title{Low-Overhead Channel Estimation Framework for Beyond Diagonal Reconfigurable Intelligent Surface Assisted Multi-User MIMO Communication}
\author{\IEEEauthorblockN{Rui Wang, Shuowen Zhang, Bruno Clerckx, and Liang Liu}
\thanks{Rui Wang, Shuowen Zhang, and Liang Liu are with the Department of Electrical and Electronic Engineering, The Hong Kong Polytechnic University, Hong Kong SAR, China (e-mails: rui-eie.wang@connect.polyu.hk, \{shuowen.zhang, liang-eie.liu\}@polyu.edu.hk).}
\thanks{Bruno Clerckx is with the Department of Electrical and Electronic Engineering, Imperial College London, London, U.K. (e-mail: b.clerckx@imperial.ac.uk).}
}
\maketitle

\begin{abstract}
    Beyond diagonal reconfigurable intelligent surface (BD-RIS) refers to a family of RIS architectures characterized by scattering matrices not limited to being diagonal and enables higher wave manipulation flexibility and large performance gains over conventional (diagonal) RIS. To achieve those promising gains, accurate channel state information (CSI) needs to be acquired in BD-RIS assisted communication systems. However, the number of coefficients in the cascaded channels to be estimated in BD-RIS assisted systems is significantly larger than that in conventional RIS assisted systems, because the channels associated with the off-diagonal elements of the scattering matrix have to be estimated as well. Surprisingly, for the first time in the literature, this paper rigorously shows that the uplink channel estimation overhead in BD-RIS assisted systems is actually of the same order as that in the conventional RIS assisted systems. This amazing result stems from a key observation: for each user antenna, its cascaded channel matrix associated with one reference BD-RIS element is a scaled version of that associated with any other BD-RIS element due to the common RIS-base station (BS) channel. In other words, the number of independent unknown variables is far less than it would seem at first glance. Building upon this property, this paper manages to characterize the minimum overhead to perfectly estimate all the channels in the ideal case without noise at the BS, and propose a two-phase estimation framework for the practical case with noise at the BS. Numerical results demonstrate outstanding channel estimation overhead reduction over existing schemes in BD-RIS assisted systems.
\end{abstract}

\begin{IEEEkeywords}
Beyond diagonal reconfigurable intelligent surface (BD-RIS), channel estimation, low-overhead communication. \blfootnote{The materials in this paper have been presented in part at the IEEE International Conference on Wireless Communications and Signal Processing (WCSP), October 2024 \cite{wcsp}.}
\end{IEEEkeywords}

\section{Introduction}

\subsection{Motivation}

Reconfigurable intelligent surfaces (RISs) have emerged as a promising technology for the sixth-generation (6G) cellular network, thanks to their unprecedented capabilities to tune the wireless propagation environment. Conventionally, the scattering matrices of RISs are diagonal matrices such that the RIS elements change the propagation properties independently \cite{Basar2019wireless,jian2022inte}. Recently, a new technology, i.e., beyond diagonal RIS (BD-RIS), has attracted more and more attention \cite{Li2023BDRIS,shen2022Modeling}. Specifically, BD-RIS leverages inter-element connections to generate non-diagonal scattering matrices. This non-diagonal structure enables more flexible wave manipulation, enhanced beamforming, and expanded coverage that are critical for 6G networks \cite{li2023beyond, Li2023BDRIS}. Because of the great potential of BD-RIS and the optimality of BD-RIS to achieve the maximum SNR and capacity of wireless channel \cite{shen2022Modeling,zheyu2025}, research has been conducted on multiple fronts: physics-consistent modeling \cite{shen2022Modeling,Nerini2024universal}, optimal and suboptimal architecture designs \cite{shen2022Modeling,Nerini2024graph,zheyu2025}, transmissive, reflective, hybrid and multi-sector modes \cite{li2023beyond,Li2023mode}, performance optimization and analysis \cite{Nerini2024closed,Santamaria2024MIMO,samy2024diagonal,björnson2025capacity}, hardware impairments \cite{Nerini2023discrete}, and integration with other technologies and applications such as integrated sensing and communication (ISAC) \cite{wang2024radar,raeisi2024localizaiton,chen2024transmitter,guang2024power}, physical layer security \cite{wang2024attacks}, full duplex systems \cite{li2024fullduplex}.

The enhanced performance of BD-RIS depends critically on the availability of channel state information (CSI). Without the BD structure, the minimum channel estimation overhead in conventional RIS assisted systems has been characterized in \cite{wang2020channel}. However, unlike conventional (diagonal) RIS with a single-connected architecture, the interconnected architecture of BD-RIS significantly increases the number of channel coefficients to be estimated. Moreover, due to the circuit requirement, the non-diagonal scattering matrices of the BD-RIS are unitary matrices \cite{li2023beyond,shen2022Modeling}. This imposes constraints since all reflecting elements must be adjusted jointly (unlike conventional RIS with independently tunable elements), thereby complicating the design of time-varying scattering strategies for channel estimation. Due to the above reasons, it is widely believed that the channel estimation overhead in BD-RIS assisted systems is significantly larger than that in conventional RIS assisted systems \cite{li2024channel,li2023channel,de2024channel,dearaujo2024semiblind,sokal2024decoupled,ginige2024prediction}. 

In this work, we consider a fully connected BD-RIS architecture \cite{shen2022Modeling} where each reflecting element is connected to all other elements to form a full non-diagonal scattering matrix. While this architecture provides the highest flexibility for wave manipulation, it inherently incurs the highest channel estimation overhead due to the largest number of interconnections compared to other BD-RIS architectures, e.g., group connected BD-RIS \cite{shen2022Modeling}, tree and forest-connected BD-RIS \cite{Nerini2024graph}, band-connected and stem-connected BD-RIS \cite{zheyu2025}, whose corresponding graphs have fewer edges (hence interconnections) than fully-connected BD-RIS. By addressing the most challenging fully connected BD-RIS architecture, we will establish a theoretical foundation relating channel estimation overhead to BD-RIS architectural complexity and provide insights that help reducing overhead for other BD-RIS architectures.

\subsection{Prior Work}

Under the BD-RIS assisted communication system, most prior works on channel estimation have focused on single-user multiple-input multiple-output (MIMO) scenarios \cite{li2024channel,li2023channel,de2024channel,dearaujo2024semiblind,sokal2024decoupled}. \cite{li2024channel} and \cite{li2023channel} proposed a least square (LS)-based method to derive a closed‑form estimation of the cascaded user‑RIS‑base station (BS) channel and formulated a BD-RIS design problem aimed at minimizing the mean square error (MSE) estimation. \cite{de2024channel} developed two tensor decomposition–based methods. One offers a closed-form solution based on the Khatri-Rao factorization (KRF) algorithm. Given an initial LS estimator of the cascaded channel, it leverages rank-one matrix approximation to extract the individual user-RIS and RIS-BS channels, yielding a noise rejection gain compared to the LS estimation scheme. The other employs an alternating LS (ALS)-based approach to directly decouple the estimation of cascaded channel into individual channel matrices, thereby reducing overhead for estimation. \cite{dearaujo2024semiblind} proposed a semi-blind joint channel and symbol estimation scheme based on tensor modeling of the received signals and a trilinear alternating estimation scheme was derived. For the multi-user scenario, \cite{ginige2024prediction} adapted the ALS method for multiple single‑antenna users by assuming that all users transmit orthogonal pilot sequences such that channels from all users can be simultaneously estimated by the BS. 

However, in a fully connected BD-RIS assisted system with $N$ BS antennas, $M$ BD-RIS elements, and $K$ users each with $U$ antennas, most of the above methods \cite{li2024channel,li2023channel,de2024channel,sokal2024decoupled} require a channel estimation overhead of $KUM^2$, because all the channel coefficients are estimated independently. Although the ALS-based method proposed in \cite{dearaujo2024semiblind,ginige2024prediction} can reduce the channel estimation overhead via 1-mode and 2-mode unfoldings of the received signal tensor, the overhead is still fundamentally high because all the channel coefficients are treated as independent variables. Note that in the conventional RIS assisted systems, it has been shown in \cite{wang2020channel} that the channel estimation overhead is $M+\lceil M(KU-1)/q\rceil$, where $q$ is the rank of RIS-BS channel. An important question motivating this work arises: \textit{Is it possible to reduce the channel estimation overhead in BD-RIS assisted systems to the same order as that in conventional RIS assisted systems?}

\subsection{Main Contributions}

In this paper, we study the channel estimation problem in BD-RIS assisted uplink communication. Specifically, our considered system consists of multiple multi-antenna users, one fully connected BD-RIS, and one multi-antenna BS. Based on the pilot signals from the users, the BS needs to estimate the channels that are necessary for the subsequent BS and BD-RIS beamforming design. We are interested in two questions. First, in the ideal case without noise at the BS, what is the minimum overhead required to perfectly estimate all the channels? Second, in the practical case with noise at the BS, how to design an efficient algorithm to accurately estimate all the channels? The contributions of this paper are summarized as follows.  
\begin{itemize}
    \item We reveal a fundamental correlation property of the cascaded user-RIS-BS channels in BD-RIS assisted systems — for each user, its cascaded channel matrix associated with a pair of BD-RIS element and BS antenna is a scaled version of that associated with any other pair of BD-RIS element and BS antenna, because they share a common RIS-BS channel. This indicates that after the cascaded channel matrix associated with a reference pair of BD-RIS element and BS antenna is estimated, each of the other cascaded channel matrices can be recovered by simply estimating a scaling number. In other words, the number of independent unknown variables in the high-dimension channels in BD-RIS assisted systems is much lower than what we expect. To our best knowledge, our work is the first one to reveal this important property in BD-RIS assisted systems. Previous works \cite{li2024channel,li2023channel,de2024channel,dearaujo2024semiblind,sokal2024decoupled,ginige2024prediction} estimate all cascaded channels independently, leading to larger channel estimation overhead.
    \item Based on the above channel property, this paper shows that in the ideal case without noise at the BS, the overhead to perfectly estimate all the cascaded channels in BD-RIS assisted systems is $2M+\lceil M(KU-1)/q \rceil$, with $q$ being the rank of RIS-BS channel matrix, which is much smaller than $KUM^2$ required by the schemes in previous works without utilizing this channel property. Moreover, the minimum channel estimation overhead in conventional RIS assisted systems is shown to be $M+\lceil M(KU-1)/q\rceil$ in \cite{wang2020channel}. Therefore, the channel estimation overhead in BD-RIS assisted systems is of the same order as that in conventional RIS assisted systems. The main difficulty to characterize the channel estimation overhead is that the received pilot signal depends non-linearly on the product of two variables needed to estimate: the cascaded channel matrix associated with a reference BD-RIS element, and the scaling coefficients associating the cascaded channel matrices of the other BD-RIS elements to this reference element. To tackle this issue, we decouple the non-linear channel estimation problem into two linear problems. First, we propose an efficient BD-RIS control method such that the received pilot signal depends only and linearly on the reference cascaded channel matrix. Second, after the reference cascaded channel matrix is estimated, the received pilot signal is a linear function of the scaling coefficients of the other cascaded channel matrices. The above approach enables us to characterize the minimum overhead to estimate all the cascaded channels based on linear equations.
    \item Moreover, we propose an efficient two-phase framework that can estimate the cascaded channels with high accuracy and low overhead in the practical case with noise at the BS. In Phase I, a reference user transmits pilot signals exclusively through its first antenna. 
    Using designed BD-RIS scattering matrices and pilot signals, we estimate the reference cascaded channel matrix and the scaling coefficients relating the other cascaded channel matrices to this reference. In Phase II, the remaining antennas — both for the reference user and other users — transmit pilot signals to estimate their respective scaling coefficients. Under the above framework, we develop linear minimum mean-squared error (LMMSE)-based estimators to estimate the reference cascaded channel and scaling coefficients. Numerical results demonstrate that given a small number of time instants, our scheme can achieve much lower channel estimation MSE compared to existing algorithms \cite{li2024channel,de2024channel}.
\end{itemize}

\subsection{Organization}

The rest of this paper is organized as follows. Section \ref{sys_model} introduces the system model. Section \ref{problem} states the channel estimation problem without and with noise at the BS, respectively, and reveals the channel property of the cascaded channels under the BD-RIS assisted system. Section \ref{special case} considers the case without noise and characterizes the minimum overhead for perfectly estimating the channels in a special case of a single-antenna user. Section \ref{general case} generalizes the result to a general case of multiple multi-antenna users. Section \ref{with noise} proposes an efficient channel estimation method under the practical case with noise. Section \ref{simulation} provides numerical examples to demonstrate the effectiveness of our proposed scheme. Section \ref{conclude} concludes the paper.

\textit{Notations}: ${\bm I}$ and ${\bm O}$ denote an identify matrix and an all-zero matrix with appropriate dimensions, respectively. For a matrix ${\bm A}$, ${\bm A}^T$ and ${\bm A}^H$ denote its transpose and conjugate transpose, respectively. For a square full-rank matrix ${\bm A}$, ${\bm A}^{-1}$ denotes its inverse. ${\rm vec}(\cdot)$ denotes the vectorization of a matrix, and ${\rm unvec}(\cdot)$ denotes the reverse operation of the vectorization. ${\bm A}(I,:)$ extracts the subset of rows from a matrix ${\bm A}$, where $I$ is a sequence of row-index. $\|\cdot\|_F$ denotes the Frobenius norm. $\lceil\cdot\rceil$ denotes the ceiling function. $\otimes$ denotes the Kronecker product. The distribution of a circularly symmetric complex Gaussian (CSCG) random vector with mean $\bm x$ and covariance matrix $\bm \Sigma$ is denoted by $\mathcal{CN}({\bm x},{\bm\Sigma})$. $\mathbb{E}[\cdot]$ denotes the expectation operator.

\section{System Model}\label{sys_model}

\begin{figure}[t]
    \centering
    \includegraphics[scale=0.1]{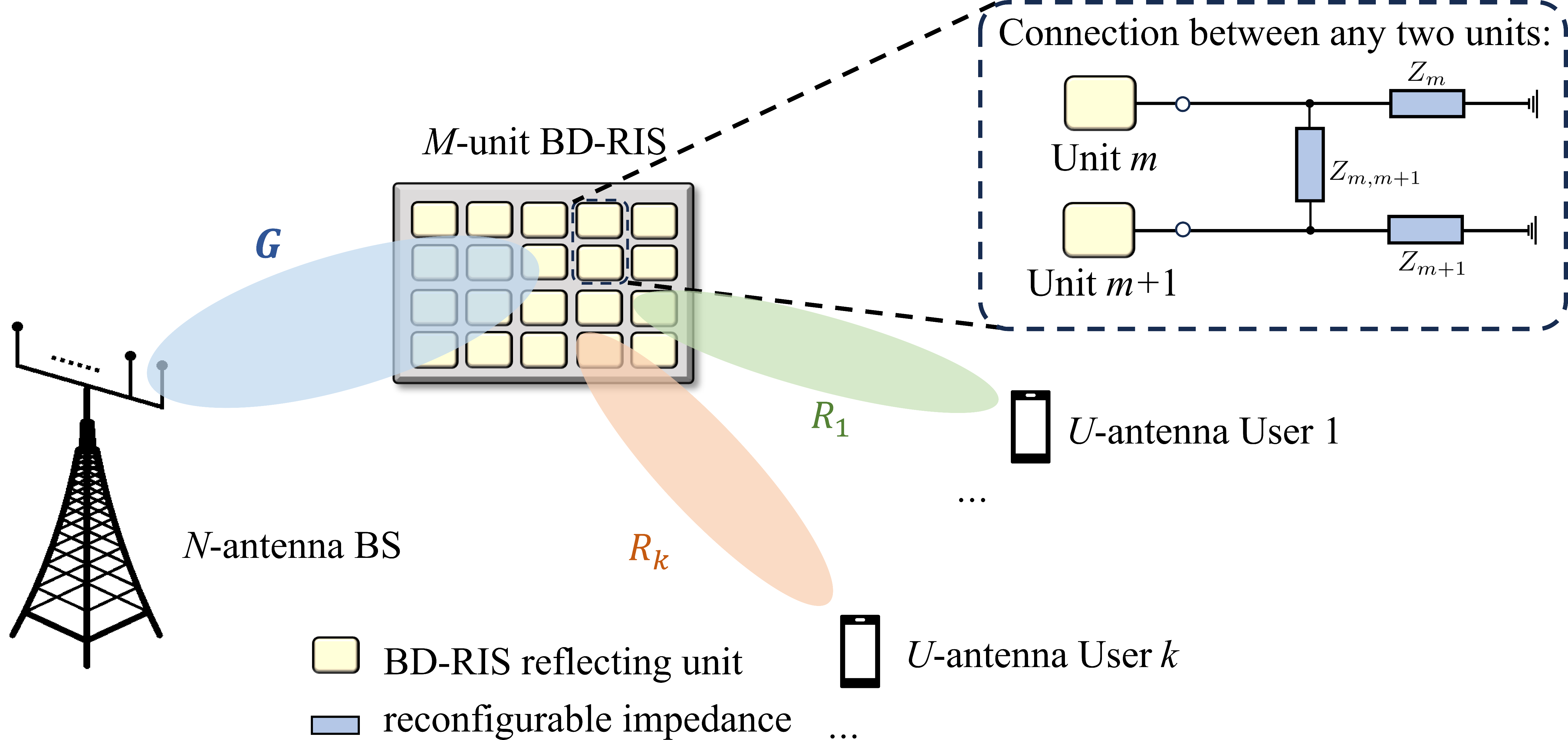}
    \caption{A BD-RIS assisted MU-MIMO uplink communication system.}
    \label{fig:sys_model}
\end{figure}

We consider an uplink communication system consisting of a BS with $N$ antennas, a BD-RIS with $M$ passive reflecting elements, and $K$ users each with $U$ antennas, as shown in Fig. \ref{fig:sys_model}. We assume a quasi-static block fading channel model, where the channels remain approximately constant in each coherence block with $T$ time instants. Define ${\bm D}_{k}\in\mathbb{C}^{N\times U}$ as the direct baseband equivalent channels from user $k$ to the BS. The baseband equivalent channel from the $u$-th antenna of user $k$ to the $m$-th BD-RIS reflecting element and that from the $m$-th BD-RIS element to the BS are denoted by $r_{k,u,m}\in\mathbb{C}$ and ${\bm g}_m\in\mathbb{C}^{N\times 1}$, $k=1,\cdots,K$, $u=1,\cdots,U$, $m=1,\cdots,M$, respectively. Define
\begin{align}\label{Rk}
    {\bm R}_k=\left[
    \begin{array}{ccc}
        r_{k,1,1} & \cdots & r_{k,U,1} \\
        \vdots & \ddots & \vdots \\
        r_{k,1,M} & \cdots & r_{k,U,M}
    \end{array}
    \right]\in\mathbb{C}^{M\times U},~~\forall k,
\end{align}
as the overall channels from user $k$ to the BD-RIS, and ${\bm G}=[{\bm g}_1,\cdots,{\bm g}_M]\in\mathbb{C}^{N\times M}$ as the overall channels from the BD-RIS to the BS. Then, at time instant $t$, the effective uplink channel from user $k$ to the BS through the BD-RIS, resulting from the direct channel and the RIS reflection channel, is expressed as
\begin{align}\label{H_ori}
    {\bm H}_{k,t}&={\bm D}_k+{\bm G}{\bm \Phi}_t{\bm R}_k,~~\forall k,t,
\end{align}
where ${\bm \Phi}_t\in\mathbb{C}^{M\times M}$ denotes the scattering matrix of the BD-RIS at time instant $t$. Since the BD-RIS has a fully connected architecture where each reflecting element is connected to the other elements, ${\bm \Phi}_t$ is a full matrix, where the entry on the $i$-th row and the $j$-th column denotes the reconfigurable coefficient of the inter-connection between the $i$-th and the $j$-th BD-RIS reflecting element. Moreover, ${\bm \Phi}_t$ is a unitary matrix, i.e.,
\begin{align}\label{unitary}
    {\bm \Phi}_t^H{\bm \Phi}_t={\bm \Phi}_t{\bm \Phi}_t^H={\bm I}_M,~~\forall t,
\end{align}
due to the circuit requirement \cite{Li2023BDRIS}. Then, the received signal of the BS at time instant $t$ is expressed as
\begin{align}\label{rev}
    {\bm y}_{t}^{\dagger}&={\sum}_{k=1}^K{\bm H}_{k,t}\sqrt{p}{\bm x}_{k,t}+{\bm n}_t \notag\\ 
    &={\sum}_{k=1}^K\left({\bm D}_k+{\bm G}{\bm \Phi}_t{\bm R}_k\right)\sqrt{p}{\bm x}_{k,t}+{\bm n}_t,~~t=1,\cdots,T,
\end{align}
where $p$ denotes the identical transmit power of all users, ${\bm x}_{k,t}\in\mathbb{C}^{U\times 1}$ denotes the unit-power transmit signal of user $k$ at time instant $t$, and ${\bm n}_t\sim\mathcal{CN}({\bm 0},\sigma^2{\bm I}_N)$ denotes the additive white Gaussian
noise (AWGN) of the BS at time instant $t$.

In this paper, we consider the legacy two-stage transmission protocol for the uplink communications, where each coherence block of length $T$ time instants is divided into the channel estimation stage consisting of $\tau<T$ time instants and data transmission stage consisting of $T-\tau$ time instant. Specifically, in Stage I, each user $k$ transmits a sequence of $\tau$ pilot symbols, i.e., ${\bm x}_{k,t}={\bm a}_{k,t}=[a_{k,1,t},\cdots,a_{k,U,t}]^T$, $k=1,\cdots,K,t=1,\cdots,\tau$, to the BS, where $a_{k,u,t}$ is the pilot symbol transmitted by the $u$-th antenna of user $k$ at time instant $t$. According to \eqref{rev}, at time instant $t$ of Stage I, the received signal of the BS is
\begin{align}\label{rev_S1}
    {\bm y}_{t}^{\dagger}={\sum}_{k=1}^K\left({\bm D}_k+{\bm G}{\bm \Phi}_t{\bm R}_k\right)\sqrt{p}{\bm a}_{k,t}+{\bm n}_t,~~t=1,\cdots,\tau.
\end{align}
The task of the BS is then to perform channel estimation based on ${\bm y}_t^{\dagger}$, $t=1,\cdots,\tau$.

In Stage II, the scattering matrix of the BD-RIS is fixed over different time instants \cite{li2024channel}, i.e., ${\bm \Phi}_{t}=\bar{\bm \Phi}$, $t=\tau+1,\cdots,T$. According to \eqref{rev}, at time instant $t$ of Stage II, the received signal of the BS is 
\begin{align}\label{rev_S2}
    {\bm y}_{t}^{\dagger}=\sum_{k=1}^K\left({\bm D}_k+{\bm G}\bar{\bm \Phi}{\bm R}_k\right)\sqrt{p}{\bm x}_{k,t}+{\bm n}_t,~~t=\tau+1,\cdots,T,
\end{align}
where ${\bm x}_{k,t}\sim\mathcal{CN}({\bm 0},{\bm S}_k)$ denotes the message signal of user $k$ at time instant $t$. Then, the channel capacity of user $k$ is
\begin{align}\label{rate_k}
    C_k=\log_2\det\left({\bm I}_N+{\bm C}_{y_k}\left({\sum}_{j\neq k}{\bm C}_{y_j}+\sigma^2{\bm I}_N\right)^{-1}\right),~~\forall k,
\end{align}
where ${\bm C}_{y_k}=p\bar{\bm H}_k{\bm S}_k\bar{\bm H}_k^H$, with $\bar{\bm H}_k={\bm D}_k+{\bm G}\bar{\bm \Phi}{\bm R}_k$ denoting the effective channel between the BS and user $k$, $\forall k$. By applying the connection between vectorization and Kronecker product \cite{henderson1981vec}, we have
\begin{align}\label{unvec}
    {\bm G}\bar{\bm \Phi}{\bm R}_k={\rm unvec}(
    {\bm J}_k
    \bar{\bm \phi}),~~\forall k,
\end{align}
where $\bar{\bm \phi}={\rm vec}(\bar{\bm \Phi})$ and
\begin{align}\label{Jk}
    {\bm J}_k={\bm R}_k^T\otimes{\bm G}=
    \left[
    \begin{array}{ccc}
        {\bm Q}_{k,1,1} & \cdots & {\bm Q}_{k,1,M}  \\
        \vdots & \ddots & \vdots \\
        {\bm Q}_{k,U,1} & \cdots & {\bm Q}_{k,U,M}
    \end{array}
    \right]
    \in\mathbb{C}^{UN\times M^2},
\end{align}
with 
\begin{align}\label{Q_kum}
    {\bm Q}_{k,u,m}=r_{k,u,m}{\bm G}\in\mathbb{C}^{N\times M},~~\forall k,u,m,
\end{align}
denoting the $(u,m)$-th sub-block of ${\bm J}_k$. Note that in the conventional RIS assisted system, since the scattering matrix $\bar{\bm \Phi}$ is diagonal, it decouples the cascaded channel associated with each user antenna into $M$ independent paths $r_{k,u,1}{\bm g}_1,\cdots,r_{k,u,M}{\bm g}_M$, which constitute each sub-block of ${\bm J}_k$ in \eqref{Jk}, i.e., ${\bm Q}_{k,u,m}'=r_{k,u,m}{\bm g}_m\in\mathbb{C}^{N\times1}$, $\forall k,u,m$. Thus, the total the number of coefficients to estimate in each ${\bm J}_k$ is $UNM$. In contrast, the non-diagonal scattering matrix of BD-RIS introduces inter-element connections, fundamentally complicating the cascaded channel structure. Specifically, each user-RIS channel component $r_{k,u,m}$ is coupled with the whole RIS-BS channel matrix ${\bm G}$ as shown in \eqref{Q_kum}, which requires $UNM^2$ channel coefficients to estimate in each ${\bm J}_k$.

\section{Problem Statement for Channel Estimation}\label{problem}

It is observed from \eqref{rate_k} and \eqref{unvec} that to maximize the network throughput by jointly optimizing the transmit covariance matrices ${\bm S}_k$'s of the users and the scattering matrix $\bar{\bm \Phi}$ at the BD-RIS, the BS has to estimate the user-BS channels ${\bm D}_k$'s and the user-RIS-BS cascaded channels ${\bm J}_k$'s in the channel estimation stage. Note that the direct user-BS channels ${\bm D}_{k}$’s can be obtained via conventional channel estimation techniques by
turning off all the RIS reflecting elements. Therefore, in this paper, we assume that ${\bm D}_k$'s are perfectly known and focus on low-overhead methods to estimate ${\bm J}_k$'s. For simplicity, under the channel estimation stage, let us define the effective received signals by removing the contribution made by the user-BS channels as:
\begin{align}\label{y_t_vec}
    {\bm y}_t&={\bm y}_t^{\dagger}-{\sum}_{k=1}^K{\bm D}_k\sqrt{p}{\bm a}_{k,t}={\sum}_{k=1}^K{\bm G}{\bm \Phi}_t{\bm R}_k\sqrt{p}{\bm a}_{k,t}+{\bm n}_t, \notag\\
    &={\sum}_{k=1}^K\sqrt{p}\left({\bm a}_{k,t}^T\otimes{\bm I}_N\right){\bm J}_k\bar{\bm \phi}_t+{\bm n}_t,~~t=1,\cdots,\tau,
\end{align}
where $\bar{\bm \phi}_t={\rm vec}({\bm \Phi}_t)$. In this paper, we consider the following two problems for estimating ${\bm J}_1,\cdots,{\bm J}_K$ based on ${\bm y}_1,\cdots,{\bm y}_{\tau}$.
\begin{itemize}
    \item Q1: In the ideal case without noise, i.e., ${\bm n}_t={\bm 0}$, $\forall t$, what is the minimum number of time instants, denoted by $\bar{\tau}$, such that ${\bm J}_1,\cdots,{\bm J}_K$ can be perfectly estimated based on ${\bm y}_1,\cdots,{\bm y}_{\bar{\tau}}$?
    \item Q2: In the practical case with noise, how to design efficient algorithms to estimate ${\bm J}_1,\cdots,{\bm J}_K$ given signals received by $\tau\geq \bar{\tau}$ time instants?
\end{itemize}

Because the received signals in \eqref{y_t_vec} are linear functions of ${\bm J}_k$'s, one straightforward way to estimate them is to solve these linear functions. Because there are $KUNM^2$ unknown variables in ${\bm J}_1,\cdots,{\bm J}_K$, theoretically speaking, the minimum number of time instants required by the above estimation method is $KUM^2$. However, the above approach is sub-optimal. This is because under the above approach, we treat the entries in each ${\bm J}_k$ as independent entries. But according to \eqref{Jk} and \eqref{Q_kum}, different sub-blocks of each ${\bm J}_k$ are highly correlated. Specifically, it follows that 
\begin{align}\label{corr}
    {\bm Q}_{k,u,m}=\beta_{k,u,m}{\bm Q}_{1,1,1}, ~~\forall (k,u,m) \neq (1,1,1),
\end{align}
where
\begin{align}\label{beta}
    \beta_{k,u,m}=\frac{r_{k,u,m}}{r_{1,1,1}}.
\end{align}
In other words, after ${\bm Q}_{1,1,1}$ is estimated, it is sufficient to estimate a scalar $\beta_{k,u,m}$ to reconstruct the whole matrix ${\bm Q}_{k,u,m}$, $\forall (k,u,m)\neq(1,1,1)$. To summarize, in ${\bm J}_k$'s, $\forall k$, independent unknown variables are ${\bm Q}_{1,1,1}$ and $\beta_{k,u,m}$'s, $\forall (k,u,m) \neq (1,1,1)$. Thus, the total number of independent unknown variables to estimate in ${\bm J}_k$'s is $NM+KMU-1$. This number is significantly smaller than the total number of unknown variables in ${\bm J}_k$'s, which is $KUNM^2$. This indicates that it is possible to estimate ${\bm J}_k$'s with much lower overhead compared to the existing schemes proposed in \cite{li2024channel} and \cite{de2024channel}, which did not exploit the channel property shown in \eqref{corr}. On the other hand, how to leverage the above channel property to address Q1 and Q2 is challenging. This is because if we treat ${\bm Q}_{1,1,1}$ and ${\beta}_{k,u,m}$'s as unknown variables, then the received signals in \eqref{y_t_vec} are no longer linear functions of these variables, which complicates the the theoretical analysis and algorithm design. In the rest of the paper, we first answer Q1 by considering a special case of a single-antenna user. This simple case can help shed the light on how the channel property shown in \eqref{corr} can significantly reduce the channel estimation overhead. Then, based on the insights from this special case, we generalize the theoretical result of Q1 to the case of multiple multi-antenna users. At last, we will design an efficient channel estimation algorithm to address Q2.

\section{Minimum Overhead in Special Case of A Single-Antenna User}\label{special case}

Let us first focus on Q1 under the special single-user and single-antenna case, i.e., $K=1$ and $U=1$. In this case, we omit the subscripts $k$ and $u$ in the variables defined above. The user-RIS channel reduces to ${\bm r}=[r_1,\cdots,r_M]^T$. Then, the cascaded channel ${\bm J}$ can be expressed as
\begin{align}\label{J_BD}
    {\bm J}&={\bm r}^T\otimes{\bm G}=[{\bm Q}_1,\cdots,{\bm Q}_M],
\end{align}
where
\begin{align}
    {\bm Q}_m=r_m{\bm G}, ~~m=1,\cdots,M,
\end{align}
is the $m$-th sub-block of $\bm J$. Similar to \eqref{corr}, the correlation among ${\bm Q}_m$'s can be modeled as:
\begin{align}\label{corr_SU}
    {\bm Q}_{m}=\beta_{m}{\bm Q}_{1}, ~~m\neq 1,
\end{align}
where
\begin{align}\label{beta_SU}
    \beta_{m}=\frac{r_{m}}{r_{1}},
\end{align}
denotes the scaling coefficient between ${\bm Q}_{m}$ and ${\bm Q}_{1}$. In the ideal case without noise, the received signal in \eqref{y_t_vec} reduces to:
\begin{align}\label{y_t_SU}
    {\bm y}_t &=\sqrt{p}{a}_{t}{\bm G}{\bm \Phi}_t{\bm r}=\sqrt{p}a_t{\bm J}\bar{\bm \phi}_t, ~~t=1,\cdots,\tau.
\end{align}
As discussed in Section \ref{problem}, to reduce the number of estimated variables, we should utilize the channel property shown in \eqref{corr_SU} and \eqref{beta_SU} to replace ${\bm J}$ by ${\bm Q}_1$ and ${\beta}_m$'s, $m=2,\cdots,M$. Then, the received signal given in \eqref{y_t_SU} reduces to 
\begin{align}\label{y_t_SU_Q1}
    {\bm y}_t=\sqrt{p}a_t\left({\bm Q}_1{\bm \phi}_{1,t}+{{\sum}}_{m=2}^M\beta_m{\bm Q}_1{\bm \phi}_{m,t}\right), ~~\forall t,
\end{align}
where ${\bm \phi}_{m,t}$ denotes the $m$-th column of ${\bm \Phi}_t$. As mentioned in Section \ref{problem}, the challenge is that ${\bm y}_t$'s given in \eqref{y_t_SU_Q1} are non-linear functions of ${\bm Q}_1$ and $\beta_m$'s. To tackle this issue, we aim to properly design the BD-RIS scattering matrices ${\bm \Phi}_t$'s and the user pilot signals $a_{t}$'s over time such that ${\bm Q}_1$ can be first estimated and $\beta_m$'s can be then estimated. Note that without the unitary constraint 
\eqref{unitary}, we can set ${\bm \phi}_{m,t}={\bm 0}$, $m=2,\cdots,M$, in the first few time instants, such that the received signal in \eqref{y_t_SU_Q1} is only contributed by the first RIS reflecting element. Therefore, we can estimate ${\bm Q}_1$ by solving linear functions. Then, in the following time instants, we can set ${\bm \phi}_{m,t}\neq{\bm 0}$, $m=2,\cdots,M$. Given ${\bm Q}_1$, the received signals given in \eqref{y_t_SU_Q1} are linear functions of ${\beta}_m$'s, and we can estimate them efficiently. However, for BD-RIS, due to the inter-connected circuits to control the reflecting elements, we cannot design a scattering matrix ${\bm \Phi}_t$ with ${\bm \phi}_{1,t}\neq{\bm 0}$ and ${\bm \phi}_{m,t}={\bm 0}$, $m=2,\cdots,M$, by shutting down the $2$nd to the $M$-th reflecting elements. Such a scattering matrix does not satisfy the unitary constraint \eqref{unitary}. In the following, we propose a new method to tackle the above problem.

Under our proposed approach, we divide the overall $\tau$ time instants into two parts with the same length $\delta$, i.e., $\tau=2\delta$. For each time instant $t\leq\delta$ and time instant $t+\delta$, we set BD-RIS scattering matrices ${\bm \Phi}_t$'s and user pilot signals according to the following rule:
\begin{itemize}
    \item Given the scattering matrix ${\bm \Phi}_t$ at time instant $t\leq\delta$ that satisfies \eqref{unitary}, we set the scattering matrix at time instant $t+\delta$ as
    \begin{align}\label{Phi_tau1T1}
        &{\bm \phi}_{1,\delta+t}=e^{j\theta}{\bm \phi}_{1,t},\notag\\
        &{\bm \phi}_{m,\delta+t}={\bm \phi}_{m,t},~~m=2,\cdots,M,
    \end{align}
    where $j$ is the imaginary unit, and $\theta$ in $(0,2\pi)$ is an arbitrary phase shift. It can be shown that the scattering matrix at time instant $t+\delta$ also satisfies \eqref{unitary}.

    \item Given the user pilot signal $a_t$ at time instant $t\leq\delta$, we set the user pilot signal at time instant $t+\delta$ to be
    \begin{align}\label{pilot_rule}
        a_{\delta+t}=a_t.
    \end{align}
\end{itemize}
Then, the received signals over $\tau=2\delta$ time instants can be re-written as
\begin{align}\label{y_t_delta}
    {\bm y}_t=\sqrt{p}a_t{\bm Q}_{1}{\bm \phi}_{1,t}
    +\sqrt{p}a_t{\sum}_{m=2}^M\beta_m{\bm Q}_{1}{\bm \phi}_{m,t}, 
\end{align}
\begin{align}\label{y_t_t+delta}
    {\bm y}_{\delta+t}=\sqrt{p}a_te^{j\theta}{\bm Q}_{1}{\bm \phi}_{1,t}
    +\sqrt{p}a_t{\sum}_{m=2}^M\beta_m{\bm Q}_{1}{\bm \phi}_{m,t}, \notag\\
    t=1,\cdots,\delta.
\end{align}
Note that in \eqref{y_t_delta} and \eqref{y_t_t+delta}, for each pair of received signals ${\bm y}_t$ and ${\bm y}_{\delta+t}$, $t=1,\cdots,\delta$, only the received signal component contributed by ${\bm Q}_{1}$ is different. As a result, by subtracting ${\bm y}_1,\cdots,{\bm y}_{\delta}$ from ${\bm y}_{\delta+1},\cdots,{\bm y}_{\tau}$, respectively, the signal components contributed by ${\bm Q}_{2}$ to ${\bm Q}_{M}$ can be eliminated, obtaining $\delta$ effective received signals merely contributed by ${\bm Q}_{1}$. The $t$-th effective received signal is 
\begin{align}
    \bar{\bm y}_t={\bm y}_{\delta+t}-{\bm y}_{t}=\sqrt{p}a_t(e^{j\theta}-1){\bm Q}_{1}{\bm \phi}_{1,t}, \notag\\
    t=1,\cdots,\delta.
\end{align}
Then, the overall effective received signal is expressed as
\begin{align}\label{Y1bar}
    \bar{\bm Y}_1=[\bar{\bm y}_{1},\cdots,\bar{\bm y}_{\delta}]=\sqrt{p}(e^{j\theta}-1){\bm Q}_{1}{\bm \Psi}_1,
\end{align}
where ${\bm \Psi}_1=[a_1{\bm \phi}_{1,1},\cdots,a_{\delta}{\bm \phi}_{1,\delta}]\in\mathbb{C}^{M\times\delta}$. $\bar{\bm Y}_1$ is a linear function of ${\bm Q}_1$. Our aim is to estimate ${\bm Q}_{1}$ based on \eqref{Y1bar} with the minimum number of time instants. This requires to find the minimum number of effective received signals $\delta$ such that ${\rm rank}({\bm \Psi}_1)=M$. It can be shown that the minimum value of $\delta$ is
\begin{align}\label{delta_min}
    \bar{\delta}=M.
\end{align}
With this value, we can design $M$ linearly independent $M\times1$ vectors ${\bm \phi}_{1,1},\cdots,{\bm \phi}_{1,M}$, to make sure ${\rm rank}({\bm \Psi}_1)=M$. Then, ${\bm Q}_{1}$ can be perfectly estimated as 
\begin{align}\label{es_Q1_per}
    {\bm Q}_{1}=(\sqrt{p}(e^{j\theta}-1))^{-1}\bar{\bm Y}_1{\bm \Psi}_1^H({\bm \Psi}_1{\bm \Psi}_1^H)^{-1}.
\end{align}

Next, we estimate the scaling coefficients $\beta_{m}$'s. We want to emphasize that after ${\bm Q}_{1}$ is estimated, with the minimum number of $\delta$ given in \eqref{delta_min}, we already have sufficient information contained in the received signals over the $\tau=2M$ time instants as given in \eqref{y_t_delta} and \eqref{y_t_t+delta} to estimate $\beta_2,\cdots,\beta_M$. Specifically, we just focus on the signals received in the first $\bar{\delta}=M$ time instants given in \eqref{y_t_delta}, because those received in the following $M$ time instants given in \eqref{y_t_t+delta} contain the same information about $\beta_m$'s. By removing the signals contributed by the first reflecting element, we obtain the following signals according to \eqref{y_t_delta}:
\begin{align}\label{tilde_y1}
    \tilde{\bm y}_t&={\bm y}_t-\sqrt{p}a_t{\bm Q}_{1}{\bm \phi}_{1,t}=\sqrt{p}a_t{\sum}_{m=2}^M\beta_{m}{\bm Q}_{1}{\bm \phi}_{m,t}, \notag\\
    &={\bm F}_t\bar{\boldsymbol{\beta}},~~t=1,\cdots,M,
\end{align}
where ${\bm F}_t=\sqrt{p}a_t{\bm Q}_{1}\left[{\bm \phi}_{2,t},\cdots,{\bm \phi}_{M,t}\right]\in\mathbb{C}^{N\times(M-1)}, t=1,\cdots,M$, $\bar{\boldsymbol{\beta}}=[\beta_{2},\cdots,\beta_{M}]^T$. Define $\tilde{\bm y}^{(1)}=[\tilde{\bm y}_{1}^T,\cdots,\tilde{\bm y}_{M}^T]^T$. It then follows that
\begin{align}\label{tilde_y^1}
    \tilde{\bm y}^{(1)}={\bm \Theta}_1\bar{\boldsymbol{\beta}},
\end{align}
where
\begin{align}\label{Theta1}
    {\bm \Theta}_1&=\left[{\bm F}_1^T,\cdots,{\bm F}_{M}^T\right]^T\in\mathbb{C}^{MN\times(M-1)}.
\end{align}
The rank of $\bm \Theta_1$ is characterized in the following theorem.
\begin{theorem}\label{theo1}
    Suppose for each time instant $t=1,\cdots,M$, ${\bm \phi}_{2,t},\cdots,{\bm \phi}_{M,t}$ are linearly independent. Then, ${\rm rank}({\bm\Theta}_1)=M-1$.
\end{theorem}
\begin{proof}
    See Appendix \ref{proof1}.    
\end{proof}
According to Theorem \ref{theo1}, the $M-1$ unknown variables in $\bar{\boldsymbol{\beta}}$ can be perfectly estimated based on the $MN$ linear equations in \eqref{tilde_y^1}. The solution is
\begin{align}\label{es_beta_per}
    \bar{\boldsymbol{\beta}}=({\bm \Theta}_1^H{\bm \Theta}_1)^{-1}{\bm \Theta}_1^H\tilde{\bm y}^{(1)}.
\end{align}
To summarize, based on our BD-RIS scattering rule and user pilot rule given in \eqref{Phi_tau1T1} and \eqref{pilot_rule}, we can utilize 
\begin{align}\label{bar_tau1}
    \bar{\tau}=2\bar{\delta}=2M,
\end{align}
time instants to perfectly estimate ${\bm Q}_1$ according to \eqref{es_Q1_per} and $\beta_2,\cdots,\beta_M$ according to \eqref{es_beta_per}. Then, we can recover ${\bm J}$ based on \eqref{corr_SU} for beamforming design. 

Via this special example, we have shown how to utilize the channel property in \eqref{corr_SU} to estimate ${\bm Q}_1$ and $\beta_2,\cdots,\beta_M$, instead of ${\bm J}$ directly as in other works \cite{li2024channel,de2024channel}. This method can reduce the channel estimation overhead from $M^2$ time instants as required by \cite{li2024channel} to $2M$. In the rest of this paper, we first generalize this minimum overhead result to the general case with multiple multi-antenna users in Section \ref{general case}. Then, we will propose an efficient channel estimation algorithm for the practical case with noise in Section \ref{with noise}.

\textit{Remark}: Note that under the conventional RIS scenario, the minimum number of time instants to estimate the cascaded channels for beamforming design in the case of a single-antenna user is $M$ \cite{mishra2019channel}. On one hand, under BD-RIS, our proposed method leads to a channel estimation overhead that is still linear in $M$, same as that in the conventional RIS scenario, although the number of unknown variables in ${\bm J}$ is quadratic in $M$. This is because our scheme can utilize the channel property shown in \eqref{corr_SU} to reduce the number of independent variables to be estimated. On the other hand, the inter-connected circuits to control different reflecting elements raise the channel estimation overhead from $M$ to $2M$. There are two reasons. First, the number of independent unknown variables is increased from $MN$ to $MN+M-1$. Second, because of the unitary constraint on scattering matrix, we have less design flexibility. We need to adopt the design rule in \eqref{Phi_tau1T1} and \eqref{pilot_rule} for creating $M$ pair of received signals, each covering two time instants.

\section{Minimum Overhead in General Case of Multiple Multi-Antenna Users}\label{general case}

In this section, we consider Q1 under the general case of multiple multi-antenna users, i.e., $K>1$ and $U>1$. In the ideal case without noise, the received signal in \eqref{y_t_vec} reduces to
\begin{align}\label{y_t_ideal}
    {\bm y}_t={\sum}_{k=1}^K{\bm G}{\bm \Phi}_t{\bm R}_k\sqrt{p}{\bm a}_{k,t}={\sum}_{k=1}^K\sqrt{p}({\bm a}_{k,t}^T\otimes{\bm I}_N){\bm J}_k\bar{\bm \phi}_t.
\end{align}
Similar to the special case in Section \ref{special case}, to reduce the number of variables to estimate, we should utilize the channel property shown in \eqref{corr} and \eqref{beta} to replace ${\bm J}_k$'s by ${\bm Q}_{1,1,1}$ and ${\beta}_{k,u,m}$'s, $(k,u,m)\neq(1,1,1)$. According to the cascaded channel given in \eqref{Jk} and the channel property in \eqref{corr} and \eqref{beta}, ${\bm J}_k$ can be re-expressed as
\begin{align}\label{Jk_kron}
    {\bm J}_k={\bm B}_k^T\otimes{\bm Q}_{1,1,1}, ~~\forall k,
\end{align}
where 
\begin{align}
    \bm{B}_k=[\boldsymbol{\beta}_{k,1},\cdots,\boldsymbol{\beta}_{k,U}]\in\mathbb{C}^{M\times U},~~\forall k,
\end{align}
with ${\bm \beta}_{k,u}=[\beta_{k,u,1},\cdots,\beta_{k,u,M}]^T$, $\forall u$, and $\beta_{1,1,1}=1$. Then, by substituting \eqref{Jk_kron} into \eqref{y_t_ideal}, the received pilot signal given in \eqref{y_t_ideal} reduces to
\begin{align}\label{y_t_Q111_Bk}
    {\bm y}_t={\sum}_{k=1}^K{\bm Q}_{1,1,1}{\bm \Phi}_t{\bm B}_k\sqrt{p}{\bm a}_{k,t}, ~~
    t=1,\cdots,\tau.
\end{align}
Similar to the special case in Section \ref{special case}, the challenge is that ${\bm y}_t$'s given in \eqref{y_t_Q111_Bk} are non-linear functions of ${\bm Q}_{1,1,1}$ and ${\bm B}_k$'s. Recall that in Section \ref{special case}, we tackled this nonlinearity by first isolating and estimating the reference cascaded channel ${\bm Q}_{1,1,1}$ adopting the design of BD-RIS scattering matrices and pilot signals as \eqref{Phi_tau1T1} and \eqref{pilot_rule}, and then estimating the corresponding scaling coefficients. By mimicking this approach, we extend it to a two-phase channel estimation protocol for the general case of multiple multi-antenna users such that ${\bm Q}_{1,1,1}$ and ${\bm B}_k$'s can be estimated separately and effectively. The overall channel estimation protocol is summarized in Fig. \ref{fig:protocol}, where in Phase I with $\tau_1<\tau$ time instants, the cascaded channel ${\bm Q}_{1,1,1}$ and scaling coefficients $\beta_{1,1,2},\cdots,\beta_{1,1,M}$ are estimated based on ${\bm y}_1,\cdots,{\bm y}_{\tau_1}$ for the $1$st antenna of user $1$; while in Phase II with $\tau_2=\tau-\tau_1$ time instants, the scaling coefficients $\beta_{k,u,m}$'s, $(k,u)\neq(1,1)$, are estimated based on ${\bm y}_{\tau_1+1},\cdots,{\bm y}_{\tau_1+\tau_2}$ for the $2$nd to $U$-th antennas of user $1$ and all antennas of users $2$ to $K$. In the following, we introduce the implementation details of Phase I and Phase II, respectively.

\begin{figure}[t]
    \centering    \includegraphics[width=1\linewidth]{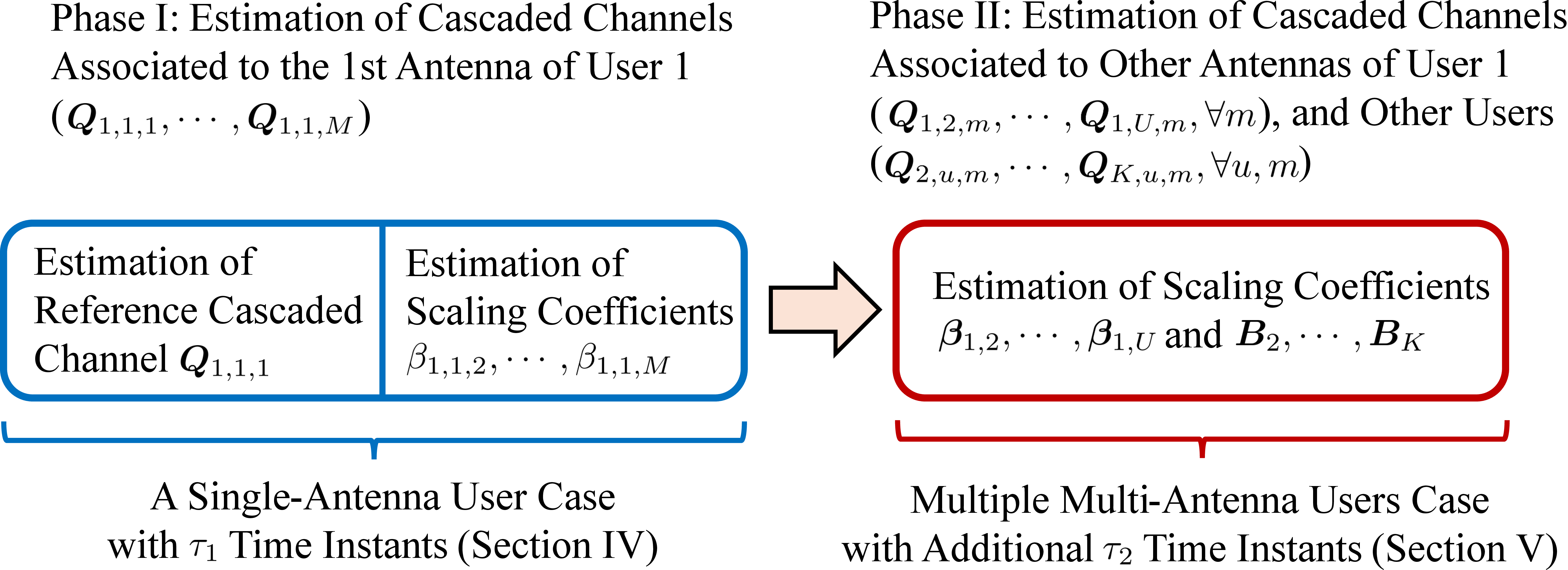}
    \caption{Illustration of the two-phase channel estimation protocol.}
    \label{fig:protocol}
\end{figure}

\subsection{Phase I: Estimation of the Cascaded Channel Associated with the 1st Antenna of User 1}\label{Phase_I}

Specifically, in Phase I at time instants $t=1,\cdots,\tau_1$, we only allow the $1$st antenna of user $1$ to transmit non-zero pilot symbols to the BS, while all the other antennas keep silent, i.e.,
\begin{align}\label{pilot_I}
    a_{k,u,t}=\begin{cases}
        a_t,~~\text{if}~k=1,u=1,\\
        0,~~~\text{otherwise},
    \end{cases}t=1,\cdots,\tau_1.
\end{align}
Then, the received signals of the BS in \eqref{y_t_Q111_Bk} over $\tau_1$ time instants can be re-written as
\begin{align}\label{y_t_P1}
    {\bm y}_t=\sqrt{p}a_t\left({\bm Q}_{1,1,1}{\bm \phi}_{1,t}+{\sum}_{m=2}^M\beta_{1,1,m}{\bm Q}_{1,1,1}{\bm \phi}_{m,t}\right), \notag\\
    t=1,\cdots,\tau_1,
\end{align}
${\bm y}_t$ given in \eqref{y_t_P1} is the same as that in the special single-user and single-antenna case in \eqref{y_t_SU_Q1}. As a result, we can apply the same method as Section \ref{special case} to estimate unknown variables of ${\bm Q}_{1,1,1}$ and $\beta_{1,1,2},\cdots,\beta_{1,1,M}$. Specifically, we divide the overall $\tau_1$ time instants in Phase I into two parts with the same length $\delta$, i.e., $\tau_1=2\delta$, and apply the same scattering matrices ${\bm \Phi}_t$'s as \eqref{Phi_tau1T1} and the same user pilot signal rule of $a_t$'s as \eqref{pilot_rule}. Thus, according to \eqref{bar_tau1} and Theorem \ref{theo1}, we can utilize 
\begin{align}
    \bar{\tau}_1=2M,
\end{align}
time instants to perfectly estimate ${\bm Q}_{1,1,1}$ and $\beta_{1,1,2},\cdots,\beta_{1,1,M}$ based on ${\bm y}_{1},\cdots,{\bm y}_{\tau_1}$ in \eqref{y_t_P1} according to \eqref{es_Q1_per} and \eqref{es_beta_per}, respectively. 

\subsection{Phase II: Estimation of the Scaling Coefficients of the Other Antennas of User 1 and the Other Users}\label{Phase_II}

After ${\bm Q}_{1,1,1}$ and $\beta_{1,1,2},\cdots,\beta_{1,1,M}$ are estimated in Phase I, the remaining unknown variables to estimate reduce to
\begin{align}\label{bar_B}
    \bar{\bm B}=[{\bm \beta}_{1,2},\cdots,{\bm \beta}_{1,U},{\bm B}_2,\cdots,{\bm B}_K]\in\mathbb{C}^{M\times(KU-1)}.
\end{align}
In the following, we introduce how to estimate $\bar{\bm B}$ in Phase II. Specifically, in Phase II at time instants $t=\tau_1+1,\cdots,\tau_1+\tau_2$, we keep the $1$st antenna of user $1$ silent, while the other antennas of user $1$ and all antennas of users $2$ to $K$ transmit non-zero pilot symbols. i.e.,
\begin{align}\label{pilot_II}
    &a_{1,1,t}=0, \notag\\
    &[a_{1,2,t},\cdots,a_{1,U,t}]^T\neq{\bm 0}, \notag\\
    &{\bm a}_{k,t}\neq{\bm 0}, ~k=2,\cdots,K, ~~t=\tau_1+1\,\cdots,\tau_1+\tau_2.
\end{align}
Then, the received signals in \eqref{y_t_Q111_Bk} over $\tau_2$ time instants can be re-written as
\begin{align}\label{y_t_P2}
    {\bm y}_t&={\bm Q}_{1,1,1}{\bm \Phi}_t\bar{\bm B}\sqrt{p}\bar{\bm a}_{t}={\bm T}_t\bar{\bm b},~~t=\tau_1+1,\cdots,\tau_1+\tau_2,
\end{align}
where $\bar{\bm a}_t=[a_{1,2,t},\cdots,a_{1,U,t},{\bm a}_{2,t}^T,\cdots,{\bm a}_{K,t}^T]^T\in\mathbb{C}^{(KU-1)\times 1}$, ${\bm T}_t=\sqrt{p}\bar{\bm a}_{t}^T\otimes{\bm Q}_{1,1,1}{\bm \Phi}_{t}\in\mathbb{C}^{N\times M(KU-1)}$, and $\bar{\bm b}={\rm vec}(\bar{\bm B})$. The overall received signal of the BS over Phase II is given by
\begin{align}\label{y^2}
    {\bm y}^{(2)}&=[{\bm y}_{\tau_1+1}^T,\cdots,{\bm y}_{\tau_1+\tau_2}^T]^T={\bm \Theta}_2\bar{\bm b},
\end{align}
where
\begin{align}\label{Theta2}
    {\bm \Theta}_2=\left[
    {\bm T}_{\tau_1+1}^T,\cdots,{\bm T}_{\tau_1+\tau_2}^T
    \right]^T\in\mathbb{C}^{N\tau_2\times M(KU-1)}.
\end{align}
The minimum value of $\tau_2$ to achieve full-rank of ${\bm \Theta}_2$ is characterized in the following theorem.

\begin{theorem}\label{theo2}
    The minimum value of $\tau_2$ to guarantee full-rank of ${\bm \Theta}_2$ is
    \begin{align}\label{bar_tau2}
        \bar{\tau}_2=\left\lceil\frac{M(KU-1)}{q}\right\rceil,
    \end{align}
    where $q$ is the rank of ${\bm G}$.
\end{theorem}
\begin{proof}
See Appendix \ref{proof2}.
\end{proof}

According to Theorem \ref{theo2}, the $M(KU-1)$ unknown variables in $\bar{\bm B}$ can be 
perfectly estimated based on the $\bar{\tau}_2N$ linear equations in \eqref{y^2}. The solution is
\begin{align}\label{es_barB_per}
    \bar{\bm b}=({\bm \Theta}_2^H{\bm \Theta}_2)^{-1}{\bm \Theta}_2^H{\bm y}^{(2)}.
\end{align}

To summarize, under the case of multiple multi-antenna users, based on \eqref{bar_tau1} and \eqref{bar_tau2}, we can utilize
\begin{align}\label{bar_tau}
    \bar{\tau}=\bar{\tau}_1+\bar{\tau}_2=2M+\left\lceil\frac{M(KU-1)}{q}\right\rceil,
\end{align}
time instants to perfectly estimate ${\bm Q}_{1,1,1}$ according to \eqref{es_Q1_per} and estimate ${\bm B}_1,\cdots,{\bm B}_K$ according to \eqref{es_beta_per} and \eqref{es_barB_per}. Then, we can recover ${\bm J}_k$'s based on \eqref{corr} for beamforming design.

\textit{Remark}: Compared to the method proposed in \cite{li2024channel}, our proposed scheme significantly reduces channel estimation overhead through the channel property indicated in \eqref{corr} from two aspects. First, the channels associated with different RIS reflecting elements for the same single antenna are highly correlated. This correlation allows us to reduce the minimum number of required time instants in Phase I from from $M^2$ to $2M$ in a single-user and single-antenna case. Second, the channels associated with different antennas and users also exhibit significant correlation. As a result, for a case of multiple multi-antenna users with a total of $KU$ antennas, the scheme in \cite{li2024channel} requires $KUM^2$ time instants, while our proposed scheme requires only $2M+\lceil\frac{M(KU-1)}{q}\rceil$ time instants.

\section{Channel Estimation for Case with Noise at the BS}\label{with noise}

In the previous section, we have shown how to perfectly estimate all the cascaded channels using $\bar{\tau}=2M+\lceil\frac{M(KU-1)}{q}\rceil$ time instants for the ideal case without noise at the BS. In this section, we consider Q2 and introduce how to estimate the cascaded channels under our proposed two-phase scheme for the practical case with noise at the BS, using $\tau\geq\bar{\tau}$ time instants.

\subsection{Phase I: Estimation of the Cascaded Channel Associated With the 1st Antenna of User 1}

In Phase I with noise at the BS, with only the $1$st antenna of user $1$ transmitting non-zero pilot symbols as \eqref{pilot_I}, the received signal of the BS given in \eqref{y_t_P1} is re-expressed as 
\begin{align}
    {\bm y}_t=\sqrt{p}a_t\left({\bm Q}_{1,1,1}{\bm \phi}_{1,t}+{\sum}_{m=2}^M\beta_{1,1,m}{\bm Q}_{1,1,1}{\bm \phi}_{m,t}\right)+{\bm n}_t, \notag\\
    t=1,\cdots,\tau_1.
\end{align}
As in Section \ref{Phase_I}, we divide the overall $\tau_1$ time instants into two durations with the same length $\delta$, i.e., $\tau_1=2\delta$. The first $\delta$ time instants consist of two parts: At time instants $t=1,\cdots,\bar{\delta}$, we set the BD-RIS scattering matrices ${\bm \Phi}_t$ and user pilot signals as in Theorem \ref{theo1}; and at time instants $t=\bar{\delta}+1,\cdots,\delta$, we set randomly generated unitary matrices ${\bm \Phi}_t$'s and random user pilot signals, i.e., $a_{t}\sim\mathcal{CN}(0,1)$. Then, for the second $\delta$ time instants $t=\delta+1,\cdots,\tau_1$, we apply the paired rule of BD-RIS scattering matrices and user pilot signals as \eqref{Phi_tau1T1} and \eqref{pilot_rule}, respectively. In this case, the received signal at time instant $t=1,\cdots,\tau_{1}$ in Phase I given in \eqref{y_t_delta} and \eqref{y_t_t+delta} is re-expressed as
\begin{align}
    {\bm y}_t=\sqrt{p}a_t{\bm Q}_{1,1,1}{\bm \phi}_{1,t}
    +\sqrt{p}a_t{\sum}_{m=2}^M\beta_{1,1,m}{\bm Q}_{1,1,1}{\bm \phi}_{m,t}+{\bm n}_t,
\end{align}
\begin{align}
    {\bm y}_{\delta+t}&=\sqrt{p}a_te^{j\theta}{\bm Q}_{1,1,1}{\bm \phi}_{1,t}
    +\sqrt{p}a_t{\sum}_{m=2}^M\beta_{1,1,m}{\bm Q}_{1,1,1}{\bm \phi}_{m,t}\notag\\
    &+{\bm n}_{\delta+t},~~t=1,\cdots,\delta,
\end{align}
respectively, where $\delta>M$. Then by subtracting ${\bm y}_1,\cdots,{\bm y}_{\delta}$ from ${\bm y}_{\delta+1},\cdots,{\bm y}_{\tau_{1}}$, respectively, we can obtain $\delta$ noisy effective received signals as follows:
\begin{align}
    \bar{\bm y}_t={\bm y}_{\delta+t}-{\bm y}_{t}=\sqrt{p}a_t(e^{j\theta}-1){\bm Q}_{1,1,1}{\bm \phi}_{1,t}+{\bm z}_t, \notag\\
    t=1,\cdots,\delta,
\end{align}
where ${\bm z}_t={\bm n}_{\delta+t}-{\bm n}_t$ denotes the effective noise, and ${\bm z}_t\sim\mathcal{CN}({\bm 0},\sigma_z^2{\bm I}_N)$ with $\sigma_z^2=2\sigma^2$. The overall effective received signal over $\tau_{1}$ time instants in Phase I is 
\begin{align}\label{Y1bar_Z}
    \bar{\bm Y}_1=[\bar{\bm y}_{1},\cdots,\bar{\bm y}_{\delta}]=\sqrt{p}(e^{j\theta}-1){\bm Q}_{1,1,1}{\bm \Psi}_1+{\bm Z}_1,
\end{align}
where ${\bm Z}_1=[{\bm z}_1,\cdots,{\bm z}_{\delta}]$. Based on \eqref{Y1bar_Z}, we apply the LMMSE estimator to estimate ${\bm Q}_{1,1,1}$ as follows:
\begin{align}\label{es_Q1}
    &\hat{\bm Q}_{1,1,1} \notag\\
    &=(\sqrt{p}(e^{j\theta}-1))^{-1}\bar{\bm Y}_1\left({\bm \Psi}_1^H{\bm C}_{Q}{\bm \Psi}_1+\sigma_z^2{\bm I}_{\delta} \right)^{-1}{\bm \Psi}_1^H{\bm C}_{Q},
\end{align}
where ${\bm C}_{Q}=\mathbb{E}[{\bm Q}_{1,1,1}^H{\bm Q}_{1,1,1}]$ denotes the covariance matrix of ${\bm Q}_{1,1,1}$.

Next, based on the estimation of ${\bm Q}_{1,1,1}$ given in \eqref{es_Q1}, the noisy version of effective received signal given in \eqref{tilde_y1} is re-expressed as
\begin{align}\label{y_step2_n}
    \tilde{\bm y}_t&={\bm y}_t-\sqrt{p}a_t\hat{\bm Q}_{1,1,1}{\bm \phi}_{1,t}={\bm F}_t\bar{\bm \beta}_{1,1}+{\bm e}_t+{\bm n}_t\notag\\
    &=\hat{\bm F}_t{\bar{\boldsymbol{\beta}}_{1,1}}+({\bm F}_t-\hat{\bm F}_t){\bar{\boldsymbol{\beta}}_{1,1}}+{\bm e}_t+{\bm n}_t, ~~t=1,\cdots,\delta,
\end{align}
where $\bar{\bm \beta}_{1,1}=[\beta_{1,1,2},\cdots,\beta_{1,1,M}]^T$, ${\bm e}_t=\sqrt{p}a_t({\bm Q}_{1,1,1}-\hat{\bm Q}_{1,1,1}){\bm \phi}_{1,t}$, and $\hat{\bm F}_t=\sqrt{p}a_t\hat{\bm Q}_{1,1,1}[{\bm \phi}_{2,t},\cdots,{\bm \phi}_{M,t}]$. Then, the overall effective received signal used for estimating $\bar{\boldsymbol{\beta}}_{1,1}$ is
\begin{align}\label{y^1+n}
    \tilde{\bm y}^{(1)}&=[\tilde{\bm y}_{1}^T,\cdots,\tilde{\bm y}_{{\delta}}^T]^T \notag\\
    &=\hat{\bm\Theta}_1\bar{\boldsymbol{\beta}}_{1,1}+({\bm\Theta}_1-\hat{\bm \Theta}_1)\bar{\boldsymbol{\beta}}_{1,1}+\tilde{\bm e}^{(1)}+\tilde{\bm n}^{(1)},
\end{align}
where $\hat{\bm \Theta}_1=[\hat{\bm F}_1^T,\cdots,\hat{\bm F}_\delta^T]^T$, $\tilde{\bm e}^{(1)}=[{\bm e}_1^T,\cdots,{\bm e}_{\delta}^T]^T$, and $\tilde{\bm n}^{(1)}=[{\bm n}_1^T,\cdots,{\bm n}_{\delta}^T]^T$. In \eqref{y^1+n}, the error propagated from Phase I, i.e., ${\bm \Theta}_1-\hat{\bm \Theta}_1$ and $\tilde{\bm e}^{(1)}$, makes it hard to obtain the LMMSE estimator of $\bar{\boldsymbol{\beta}}_{1,1}$. In practice, we can increase the pilot sequence length in Phase I, i.e, $\tau_1$, such that ${\bm \Theta}_1-\hat{\bm \Theta}_1$ and $\tilde{\bm e}^{(1)}$ are sufficiently small. In this case, we assume that ${\bm \Theta}_1-\hat{\bm \Theta}_1\approx{\bm 0}$ and $\tilde{\bm e}^{(1)}\approx{\bm 0}$. Then, \eqref{y^1+n} reduces to
\begin{align}\label{y^1_approx}
    \tilde{\bm y}^{(1)}\approx\hat{\bm \Theta}_1{\bar{\bm \beta}}_{1,1}+\tilde{\bm n}^{(1)}.
\end{align}
Based on \eqref{y^1_approx}, the LMMSE estimator of ${\bar{\bm \beta}}_{1,1}$ is designed as
\begin{align}\label{es_beta11}
    \hat{\boldsymbol{\beta}}_{1,1}={\bm C}_{\beta_1}\hat{\bm \Theta}_1^H\left(\hat{\bm \Theta}_1{\bm C}_{\beta_1}\hat{\bm \Theta}_1^H+\sigma^2{\bm I}_{N\delta}\right)^{-1}\tilde{\bm y}^{(1)},
\end{align}
where ${\bm C}_{\beta_1}$ denotes the covariance matrix of ${\bar{\bm \beta}}_{1,1}$.

\subsection{Phase II: Estimation of the Scaling Coefficients of the Other Antennas of User 1 and the Other Users}

In Phase II with noise at the BS, with the pilot transmission rule in \eqref{pilot_II} and the estimation of ${\bm Q}_{1,1,1}$ given in \eqref{es_Q1}, the received signal of the BS given in \eqref{y_t_P2} is re-expressed as
\begin{align}
    {\bm y}_t&={\bm T}_t\bar{\bm b}+{\bm n}_t \notag\\
    &=\hat{\bm T}_t\bar{\bm b}+({\bm T}_t-\hat{\bm T}_t)\bar{\bm b}+{\bm n}_t,~~t=\tau_1+1,\cdots,\tau_1+\tau_2,
\end{align}
where $\hat{\bm T}_t=\sqrt{p}\bar{\bm a}_{t}^T\otimes\hat{\bm Q}_{1,1,1}{\bm \Phi}_{t}$. Then, the overall received signal for estimating $\bar{\bm b}$ is re-expressed as
\begin{align}\label{y^2+n}
    {\bm y}^{(2)}={\bm \Theta}_2\bar{\bm b}+{\bm n}^{(2)}=\hat{\bm \Theta}_2\bar{\bm b}+({\bm \Theta}_2-\hat{\bm \Theta}_2)\bar{\bm b}+{\bm n}^{(2)},
\end{align}
where $\hat{\bm \Theta}_2=[\hat{\bm T}_{\tau_1+1}^T,\cdots,\hat{\bm T}_{\tau_1+\tau_2}^T]^T$, and ${\bm n}^{(2)}=[{\bm n}_{\tau_1+1}^T,\cdots,{\bm n}_{\tau_1+\tau_2}^T]^T$. Similar to Phase I, we can increase the pilot sequence length in Phase II, i.e., $\tau_2$, such that the propagated error ${\bm \Theta}_2-\hat{\bm \Theta}_2$ can be sufficiently small. In this case, we assume that ${\bm \Theta}_2-\hat{\bm \Theta}_2\approx{\bm 0}$. Then, \eqref{y^2+n} reduces to
\begin{align}\label{y^2_approx}
    {\bm y}^{(2)}\approx\hat{\bm\Theta}_2\bar{\bm b}+{\bm n}^{(2)}.
\end{align}
To estimate $\bar{\bm b}$ based on \eqref{y^2_approx}, we set the scattering matrices and user pilot signals as follows. At time instants $t=\tau_1+1,\cdots,\tau_1+\bar{\tau}_2$, we set the BD-RIS scattering matrices ${\bm \Phi}_t$ and user pilot signals as in Theorem \ref{theo2}; and at time instants $t=\tau_1+\bar{\tau}_2+1,\cdots,\tau_1+\tau_2$, we set randomly generated unitary matrices ${\bm \Phi}_t$'s and random user pilot signals, i.e., $a_{k,u,t}\sim\mathcal{CN}(0,1)$, $\forall k,u$. Then, based on \eqref{y^2_approx}, the LMMSE estimator of $\bar{\bm b}$ can be designed as
\begin{align}\label{es_beta}
    \hat{{\bm b}}={\bm C}_{b}\hat{\bm \Theta}_2^H\left(\hat{\bm \Theta}_2{\bm C}_{b}\hat{\bm \Theta}_2^H+\sigma^2{\bm I}_{N\tau_2}\right)^{-1}{\bm y}^{(2)},
\end{align}
where ${\bm C}_{b}$ denotes the covariance matrix of ${\bar{\bm b}}$. 

In summary, with the estimation of  $\hat{\bm Q}_{1,1,1}$ given in \eqref{es_Q1} and the estimations of $\beta_{k,u,m}$'s given in \eqref{es_beta11} and \eqref{es_beta}, the cascaded channels associated with the $m$-th BD-RIS reflecting element and the $u$-th antenna of user $k$ can be estimated as
\begin{align}
    &\hat{\bm Q}_{k,u,m}=\hat{\beta}_{k,u,m}\hat{\bm Q}_{1,1,1},~~\forall (k,u,m)\neq(1,1,1),
\end{align}
based on \eqref{corr}.

\section{Numerical Results}\label{simulation}

In this section, we provide numerical results to demonstrate the advantages of our proposed scheme. The channel between the BS and the BD-RIS and that between the $u$-th antenna of user $k$ and the $m$-th BD-RIS reflecting element are modeled as ${\bm G}\sim\mathcal{CN}({\bm 0},\ell^{\rm RB}M{\bm I})$ and $r_{k,u,m}\sim\mathcal{CN}(0,\ell^{\rm UR})$, respectively, where $\ell^{\rm RB}$ and $\ell^{\rm UR}$ denote the pass loss and follow the same model as \cite{wang2020channel}. In this case, the rank of ${\bm G}$ is $q=\min\{M,N\}$. The transmit power of users is $p=33$ dBm. The power spectrum density of the noise at the BS is assumed to be $-169$ dBm/Hz, and the channel bandwidth is $1$ MHz. The normalized mean-squared error (NMSE) is used as the metric to evaluate the performance of channel estimation. Specifically, the overall NMSE for estimating all the users' cascaded channels is defined as
\begin{align}
    {\rm NMSE}=\mathbb{E}\left[\frac{1}{K}{\sum}_{k=1}^K\frac{||\hat{\bm J}_{k}-{\bm J}_{k}||_F^2}{||{\bm J}_{k}||_F^2}\right],
\end{align}
where the overall estimated channel of user $k$, i.e., $\hat{\bm J}_k$, is given as the same form as ${\bm J}_k$ in \eqref{Jk}, with ${\bm Q}_{k,u,m}$ replaced by $\hat{\bm Q}_{k,u,m}$, $\forall k,u,m$.

We first explore the effect of different allocations between the pilot sequence length of two phases under our proposed scheme. Define $\tau_{\rm res}=\tau-\bar{\tau}$ as the length of pilot sequence exceeding the minimum value $\bar{\tau}$ in \eqref{bar_tau}, and $\varrho$ as the proportion of the additional pilot sequence length $\tau_{\rm res}$ allocated to Phase I. In Fig. \ref{fig3}, we set the numbers of BD-RIS elements, antennas at the BS, antennas at the user, and the number of users as $M=8$, $N=4$, $U=2$, and $K=1$, respectively. The total pilot sequence length $\tau$ is set as $30$, $50$, $80$ and $130$, respectively. We show the NMSE performance of our proposed scheme under different values of $\varrho$. It is observed that for each set of total pilot sequence length $\tau$, the NMSE
shows a “first-drop-then-rise”trend. Note that as $\varrho$ increases, ${\bm Q}_{1,1,1}$ can be estimated based on more received signals in Phase I, leading to less estimation error propagated to Phase II, which, however, reduces the number of pilot symbols available for Phase II to estimate $\bar{\bm B}$. When the pilot sequence length in Phase I is small, the estimation error propagated to Phase I is the bottleneck to limit the overall NMSE performance, and it is beneficial to increase $\varrho$. However, when the pilot sequence length in Phase I is sufficiently large, the BS has enough pilot signals to estimate ${\bm Q}_{1,1,1}$, and it is not a good idea to keep increasing $\varrho$ because this will reduce the pilot signals to estimate $\bar{\bm B}$ in Phase II. This indicates that the pilot sequence length allocation should be
carefully designed. In the rest of this section, we always set $\varrho$ such that an optimal NMSE performance can be achieved under our proposed scheme. 

\begin{figure}[t]
    \centering
    \includegraphics[width=1.0\linewidth]{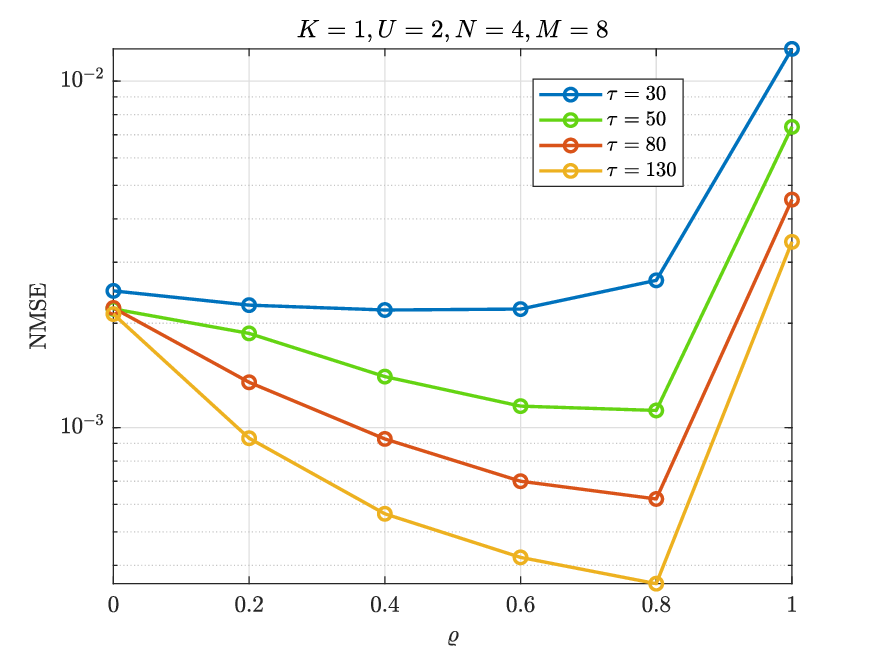}
    \caption{NMSE performance under different pilot sequence allocations.}
    \label{fig3}
\end{figure}

\subsection{Single-User Case}\label{single-user}

In the literature, the channel estimation problem under BD-RIS assisted network has also been studied in \cite{li2024channel} and \cite{de2024channel}. Different from our work which utilizes the hidden channel property given in \eqref{corr}, the above works directly estimate all entries in ${\bm J}_k$'s based on \eqref{y_t_vec}. To show the performance gain of our proposed scheme, we adopt the following benchmark schemes.
\begin{itemize}
    \item \textbf{Benchmark Scheme I}: The LS estimator proposed in \cite{li2024channel}, which directly applies the LS technique to estimate ${\bm J}_k$ based on \eqref{y_t_vec}.
    \item \textbf{Benchmark Scheme II}: The Block Tucker Kronecker factorization (BTKF) algorithm \cite{de2024channel}, which obtains a LS channel estimate via the 3‑mode unfolding of
    the received pilot tensor and then yields decoupled estimates using rank‑one (Kronecker) approximation.
    \item \textbf{Benchmark Scheme III}: The Block Tucker alternating least squares (BTALS) algorithm\cite{de2024channel}, which iteratively refines decoupled channel estimates via alternating LS on the 1‑mode and 2‑mode tensor unfoldings of the received pilot tensor.
\end{itemize}

\begin{figure}[t]
    \centering
    \includegraphics[width=1.0\linewidth]{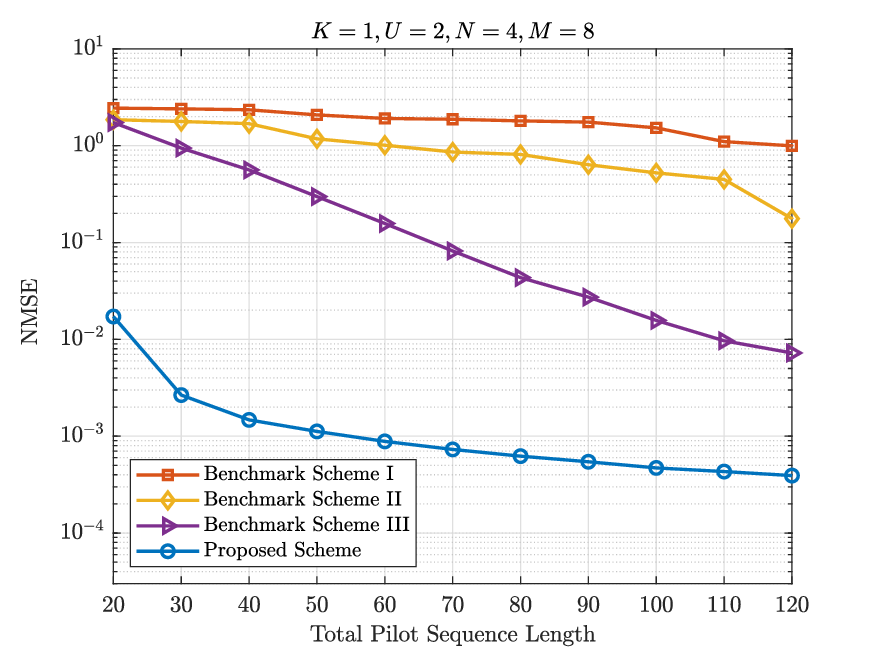}
    \caption{NMSE performance versus total pilot length when $M=8,N=4,U=2$.}
    \label{fig4}
\end{figure}

Figs. \ref{fig4}-\ref{fig6} show the NMSE performance comparison between our proposed scheme and the three benchmark schemes under the single-user case. In Fig. \ref{fig4}, we set the numbers of BD-RIS elements, antennas at the BS and antennas at the user as $M=8$, $N=4$ and $U=2$, respectively, and the total pilot sequence length ranges from $20$ to $120$. It is observed that our proposed scheme shows a significant NMSE performance gain compared to the three benchmark schemes. For Benchmark Scheme I and Benchmark Scheme II, this is because the minimum pilot sequence length required by these two schemes is $128$ while that required by our proposed scheme is $\bar{\tau}=18$ due to the utilization of \eqref{corr}. Note that Benchmark Scheme III achieves a lower NMSE than Benchmark Scheme I and II, due to that it already exploits the tensor decomposition structure of the cascaded channel to reduce estimation overhead. However, Benchmark Scheme III still performs worse than our proposed scheme because it treats all the channel coefficients as independent variables.

\begin{figure}[t]
    \centering
    \includegraphics[width=1.0\linewidth]{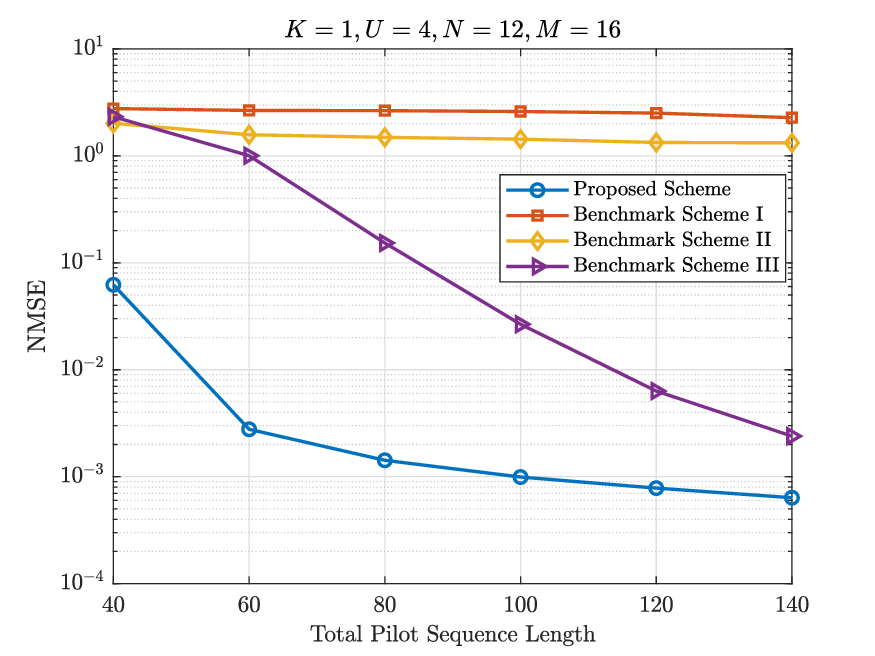}
    \caption{NMSE performance versus number of BD-RIS elements when $M=16,N=12,U=4$.}
    \label{fig5}
\end{figure}

\begin{figure}[t]
    \centering
    \includegraphics[width=1.0\linewidth]{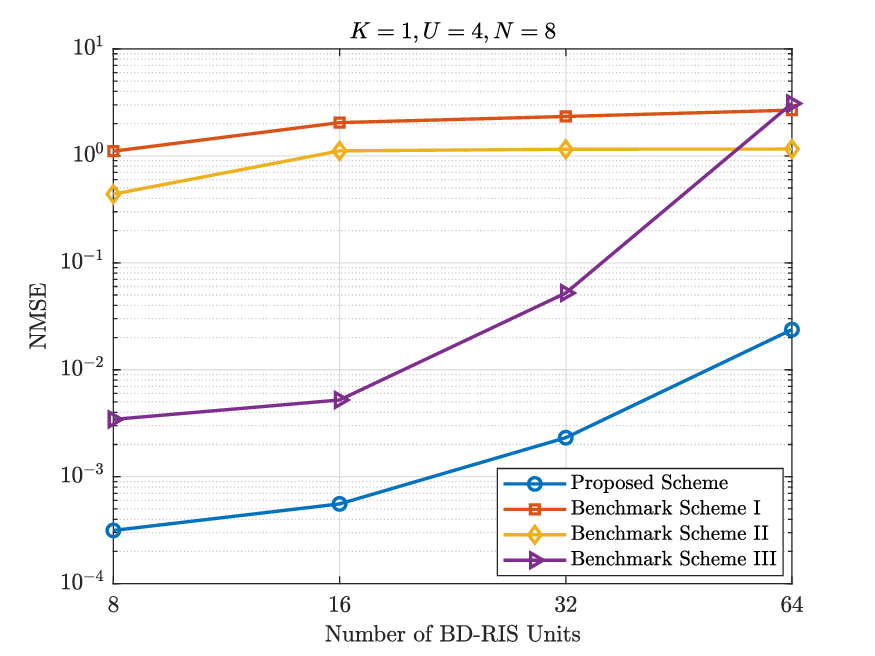}
    \caption{NMSE performance versus number of BD-RIS elements when $N=8, U=4, \tau=160$.}
    \label{fig6}
\end{figure}

Fig. \ref{fig5} shows another numerical example with larger numbers of BD-RIS elements, antennas at the BS and at the user with $M=16, N=12$ and $U=4$. The total pilot sequence length ranges from $40$ to $140$. Observations similar to Fig. \ref{fig4} can be obtained. Under this setup, the minimum pilot sequence length under our proposed scheme is $36$ while that under Benchmark Scheme I and Benchmark Scheme II is $1024$. Thus, these two schemes can hardly work. Moreover, note that in both Fig. \ref{fig4} and \ref{fig5}, the performance gain of our proposed scheme compared to Benchmark Scheme III becomes smaller as the pilot length increases. This is because Benchmark Scheme III exploits an iterative alternating LS strategy that gradually refines the channel estimates and averages noise, especially with an abundance of pilot symbols \cite{kolda2009tensor,cichocki2009nonnegative}. However, this improved performance comes at the expense of higher pilot transmission overhead due to its not exploiting the correlation in \eqref{corr}, which aggravates its disadvantage of high computational complexity caused by iterative nature.

Fig. \ref{fig6} shows the NMSE performance versus the number of the BD-RIS elements, with $N=8$ and $U=4$. The total pilot sequence length is fixed as $160$. It is observed that Benchmark Scheme I and Benchmark Scheme II show a deteriorated NMSE performance, which is caused by the pilot sequence length being insufficient, i.e., less than the minimum required
value. Moreover, when the number of BD-RIS elements increases, the NMSE increases more significantly under Benchmark Scheme III than our proposed scheme. This is attributed to that exploiting \eqref{corr} is more efficient than utilizing tensor decomposition in channel estimation overhead reduction. Specifically, for the single-user and single-antenna case, every additional BD-RIS element increases only $1$ unknown variable to estimate under our proposed scheme, while it causes $N+1$ unknown variables to estimate under Benchmark Scheme III. 

\subsection{Multi-User Case}

In the literature, no work has considered the multi-user and multi-antenna case yet. To show the performance superiority of our proposed scheme, in this subsection, we extend the benchmark schemes in Section \ref{single-user} to the multi-user case, where users transmit pilot signals at orthogonal time instants. The benchmark schemes are as follows.
\begin{itemize}
    \item \textbf{Benchmark Scheme I}: The LS estimator is applied to estimate all users' cascaded channels ${\bm J}_k$'s.
    \item \textbf{Benchmark Scheme II}: The BTKF algorithm is applied to estimate all users' cascaded channels ${\bm J}_k$'s.
    \item \textbf{Benchmark Scheme III}: The BTALS algorithm is applied to estimate all users' cascaded channels ${\bm J}_k$'s.
\end{itemize}

\begin{figure}[t]
    \centering
    \includegraphics[width=1.0\linewidth]{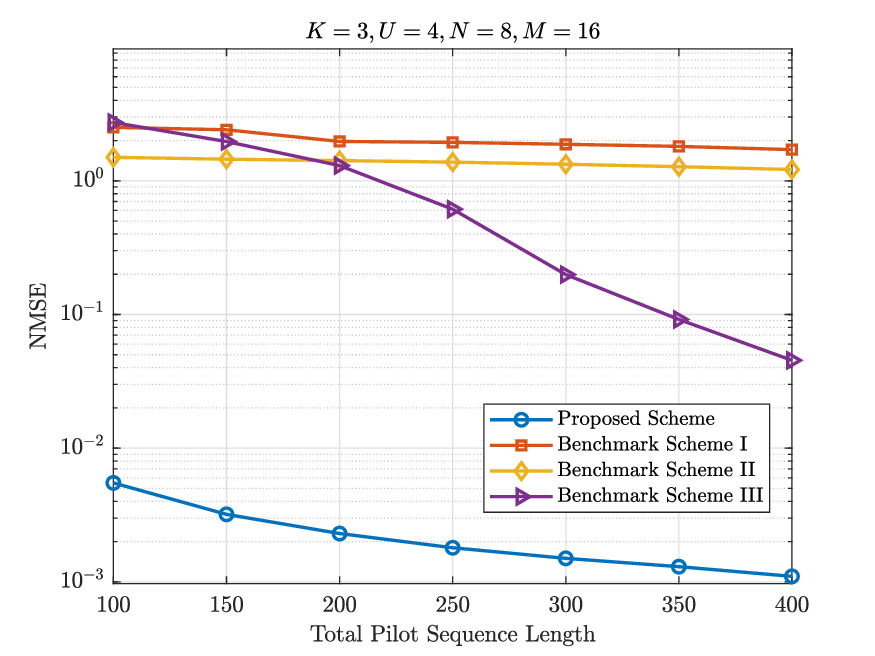}
    \caption{NMSE performance versus total pilot length when $M=16,N=8,U=4,K=3$.}
    \label{fig7}
\end{figure}

\begin{figure}[t]
    \centering
    \includegraphics[width=1.0\linewidth]{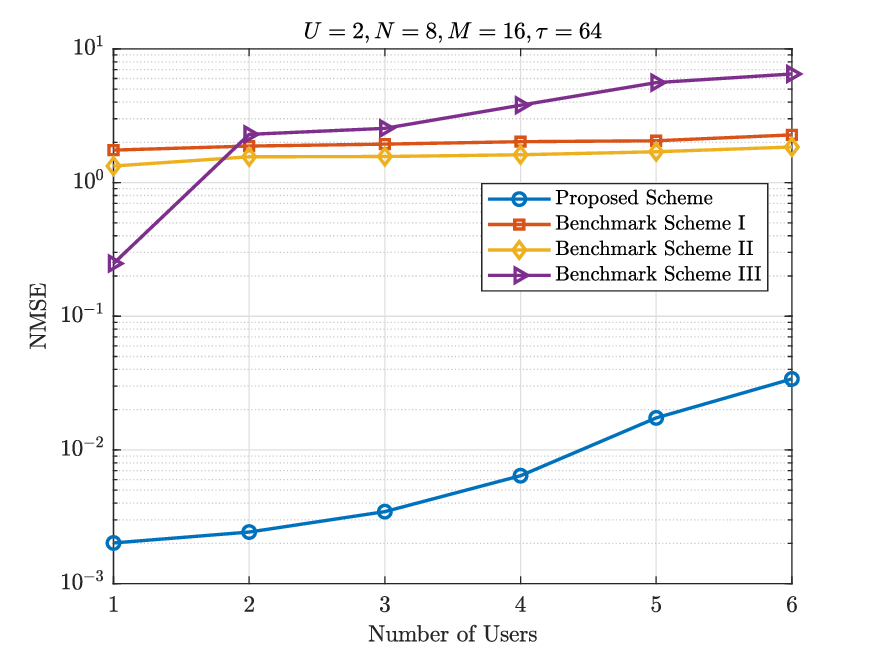}
    \caption{NMSE performance versus number of users when $M=16,N=8,U=2$.}
    \label{fig8}
\end{figure}

In Fig. \ref{fig7}, the numbers of BD-RIS elements, antennas at the BS and users, and users are set as $M=16,N=8,U=4,K=3$, respectively. The total pilot sequence length ranges from $100$ to $400$. It is observed that our proposed scheme shows a significant NMSE performance gain compared to the three benchmark schemes, due to leveraging the channel properties shown in \eqref{corr}. Note that although Benchmark Scheme III achieves a lower NMSE than Benchmark Scheme I and Benchmark Scheme II due to exploiting the built-in tensor decomposition structure
of the cascaded channel, it regards the channels of different users as independent and performs much worse than our proposed scheme in the multi-user case.

Fig. \ref{fig8} shows the NMSE performance versus the number of users, where $M=16,N=8,U=2$, the total pilot length is fixed as $\tau=64$. It is observed that as the number of users increases, the NMSE performance of all three benchmark schemes degrades markedly. In contrast, our proposed scheme maintains robust NMSE performance even with such a limited pilot sequence length, indicating its superiority in the multi-user case.

\section{Conclusions}\label{conclude}

In this paper, we revealed an important channel property for BD-RIS assisted multi-user MIMO communications. Specifically, arising from the common RIS-BS channel, the cascaded channel associated with one pair of BD-RIS element and BS antenna is a scaled version of the associated with any other pair of BD-RIS element and BS antenna. By exploiting this correlation, we addressed two key challenges: quantifying the minimum pilot transmission overhead for perfect channel estimation in noise-free case and proposing practical estimation design in the case with noise. Interestingly, it is theoretically shown that the overhead for channel estimation in BD-RIS assisted systems is in the same order as that in conventional RIS assisted systems, although the number of channel coefficients in these two systems are different by orders of magnitude. Numerical results verified that our approach can estimate channels accurately with significantly reduced overhead compared to existing algorithms.

\begin{appendix}

\subsection{Proof of Theorem \ref{theo1}}\label{proof1}
Denote the rank of ${\bm G}$ as $q$. Then, ${\rm rank}({\bm Q}_{1})={\rm rank}({\bm G})=q$. The singular value decomposition (SVD) of ${\bm Q}_{1}$ is expressed as ${\bm Q}_{1}={\bm{U\Sigma V}}^H$, where $\bm U$ and $\bm V$ are $N\times N$ and $M\times M$ unitary matrix, respectively, and ${\bm \Sigma}$ is a $N\times M$ rectangular diagonal matrix with non-zero singular values sorted in descending order on the main diagonal, i.e., $\sigma_1 \geq \cdots \geq \sigma_q$ . By left-multiplying each ${\bm F}_t$ with  ${\bm U}^H$, a new matrix with the same rank as ${\bm \Theta}_1$ can be obtained as
\begin{align}
    \tilde{\bm \Theta}_1&=\left[({\bm U}^H{\bm F}_1)^T,\cdots,({\bm U}^H{\bm F}_M)^T\right]^T\notag\\
    &=\left[ ({\bm \Sigma}\tilde{\bm \Phi}_1)^T,\cdots,({\bm \Sigma}\tilde{\bm \Phi}_M)^T \right]^T \in\mathbb{C}^{MN\times(M-1)},
\end{align}
where $\tilde{\bm \Phi}_t=\sqrt{p}a_t{\bm V}^H[{\bm \phi}_{2,t},\cdots,{\bm \phi}_{M,t}]$, and ${\bm \Sigma}\tilde{\bm \Phi}_t\in\mathbb{C}^{N\times(M-1)}$ has $q$ non-zero rows, $t=1,\cdots,M$. We can construct ${\bm \Phi}_t$'s as follows. Define $\epsilon=\lfloor\frac{M-1}{q}\rfloor$ and $\rho=M-1-\epsilon q$. The first $q$ rows of ${\bm \Phi}_t$, $t=1,\cdots,\epsilon$, are from the $[(t-1)q+1]$-th to the $tq$-th rows of an arbitrary $M$ by $M$ unitary matrix $\bm P$. Next, the remaining $M-q$ rows are complemented by random vectors, obtaining a matrix ${\bm \Gamma}$. Then, we perform QR decomposition on ${\bm \Gamma}$, obtaining the unitary matrix ${\bm \Phi}_t^T$. The transpose of ${\bm \Phi}_t^T$ is the final matrix ${\bm \Phi}_t$. Last, ${\bm \Phi}_{\epsilon+1}$ is constructed starting with the first $\rho$ rows, which are from the $(\epsilon q+1)$-th to the $(M-1)$-th rows of ${\bm P}$. Then, the remaining rows of ${\bm \Phi}_{\epsilon+1}$ are completed by the same process of constructing ${\bm \Phi}_{t}$'s, $t=1,\cdots,\epsilon$. In this case, since ${\bm \phi}_{2,t},\cdots,{\bm \phi}_{M,t}$ are linearly independent, $[{\bm \phi}_{2,t},\cdots,{\bm \phi}_{M,t}]$ is a matrix with full-rank of $M-1$. Then, $\tilde{\bm \Phi}_t$ is also rank $M-1$ for ${\bm V}$ is a unitary matrix. Thus, the $q$ non-zero rows in ${\bm \Sigma}\tilde{\bm \Phi}_t$'s, $t=1,\cdots,\epsilon$, and the $\rho$ non-zero rows in ${\bm \Sigma}\tilde{\bm \Phi}_{\epsilon+1}$ are linearly independent. Therefore, ${\rm rank}({\bm \Theta}_1)={\rm rank}(\tilde{\bm \Theta}_1)=\epsilon q+\rho=M-1$. Note that the minimum number of time instants to achieve full-rank ${\bm \Theta}_1$ is actually
\begin{align}\label{bar_delta2}
    \bar{\gamma}=\left\lceil\frac{M-1}{q}\right\rceil=\epsilon+1.
\end{align}
Since $M>\bar{\gamma}$, ${\rm rank}({\bm \Theta}_1)=M-1$ already holds.
\end{appendix}

\subsection{Proof of Theorem \ref{theo2}}\label{proof2}
Since ${\bm \Phi}_t$ is a unitary matrix, ${\rm rank}({\bm Q}_{1,1,1}{\bm \Phi}_t)={\rm rank}({\bm Q}_{1,1,1})=q$. Based on the SVD of ${\bm Q}_{1,1,1}$, we have ${\bm Q}_{1,1,1}{\bm \Phi}_t={\bm U}{\bm \Sigma}{\bm V}^H{\bm \Phi}_t$. Define ${\bm \Xi}_{t}={\bm \Sigma}{\bm V}^H{\bm \Phi}_t\in\mathbb{C}^{N\times M}$, $t=\tau_1+1,\cdots,\tau_1+\tau_2$, then ${\bm \Xi}_{t}$ has a rank of $q$ with the first $q$ rows being non-zero. By left-multiplying each ${\bm Q}_{1,1,1}{\bm \Phi}_t$ with ${\bm U}^H$, we can obtain a new matrix with the same rank as ${\bm \Theta}_2$ as
\begin{align}
    \tilde{\bm \Theta}_2&=({\bm I}_{\tau_2}\otimes{\bm U}^H){\bm \Theta}_2=\sqrt{p}\left[
    \begin{array}{c}
        \bar{\bm a}_{\tau_1+1}^T\otimes{\bm \Xi}_{\tau_1+1}  \\
         \cdots \\
         \bar{\bm a}_{\tau_1+\tau_2}^T\otimes{\bm \Xi}_{\tau_1+\tau_2}  
        \end{array}
        \right].
    \end{align}
Based on the rank property of the Kronecker product, ${\rm rank}(\bar{\bm a}_t^T\otimes{\bm \Xi}_t)={\rm rank}(\bar{\bm a}_t^T){\rm rank}({\bm \Xi}_t)=q$. The largest row rank of $\tilde{\bm \Theta}_2$ is thus $\tau_2q$ and the largest column rank is $M(KU-1)$. As a result, as long as $\tau_2\geq\lceil\frac{M(KU-1)}{q}\rceil$, then $\tilde{\bm \Theta}_2$ can achieve full rank of $M(KU-1)$.
    
Next, we show that if $\tau_2=\lceil\frac{M(KU-1)}{q}\rceil$, there always exists a design strategy of the pilot signals $\bar{\bm a}_t$'s and the BD-RIS scattering matrices ${\bm \Phi}_t$'s, $t=\tau_1+1,\cdots,\tau_1+\tau_2$, such that $\tilde{\bm \Theta}_2$ can achieve full rank. In this case, denote $\tilde{U}=KU-1$ and re-write $\bar{\bm a}_t$ as $\bar{\bm a}_t=[\bar{a}_{1,t},\cdots,\bar{a}_{\tilde{U},t}]^T$, then $\tilde{\bm \Theta}_2$ can be equivalently expressed as
\begin{align}\label{Theta_2_Xi}
    \tilde{\bm\Theta}_2 = \left[
    \begin{array}{ccc}
        \bar{a}_{1,\tau_1+1}\bm{\Xi}_{\tau_1+1}  & \cdots &  \bar{a}_{\tilde
           U, \tau_1+1}\bm{\Xi}_{\tau_1+1} \\
           \vdots  & \ddots & \vdots \\
           \bar{a}_{1,\tau_1+\tau_2}\bm{\Xi}_{\tau_1+\tau_2} & \cdots &  \bar{a}_{\tilde
           U,\tau_1+\tau_2}\bm{\Xi}_{\tau_1+\tau_2}
    \end{array} \right].
\end{align}
Based on \eqref{Theta_2_Xi}, $\tilde{\bm \Theta}_2$ is divided into $\tau_2\times\tilde{U}$ blocks, where $\tau_2$ is the number of time instants and $\tilde{U}$ represents the total number of  antennas. Our strategy is to construct an full-rank matrix $\tilde{\bm \Theta}_2$  by assembling carefully designed block matrices ${\bm \Xi}_t$'s.

To achieve this aim, we first build an auxiliary matrix ${\bm \Pi}$ which indicates the positions where $\bar{a}_{u,t}$'s are set to $1$ in $\tilde{\bm \Theta}_2$. The entries of $\bm \Pi$ serve the following purposes:
\begin{itemize}
    \item Row Constraint: Each row (corresponding to a time instant) must be allocated a number of $q$ ``supplies".
    \item Column Constraint: Each column (corresponding to an antenna) requires a number of $M$ ``demands".
\end{itemize}
These constraints ensure that over $\tau_2$ time instants we accumulate a total of $\tau_2q$ linearly independent rows and $M\tilde{U}$ linearly independent columns in $\tilde{\bm \Theta}_2$. The process of building $\bm \Pi$ draws inspiration from the Northwest Corner Rule, a classical method used in transportation problems to generate an initial feasible solution\cite{henderson1981vec}, referred to Algorithm \ref{alg1}. 

\begin{algorithm}[t]
\caption{Auxiliary Matrix Generation via the Northwest Corner Rule}\label{alg1}
\begin{algorithmic}[1]
    \Statex \textbf{Input:} $\tau_2, q, M, \tilde{U}$
    \Statex \textbf{Output:} auxiliary matrix ${\bm \Pi}$
    \State $i \gets 1$, $j \gets 1$
    \State ${\bm \Pi} \gets {\bm O}_{\tau_2\times \tilde{U}}$,  ${\textbf{sup}} \gets q{\bm 1}_{\tau_2\times 1}$, ${\textbf{dem}} \gets M{\bm 1}_{\tilde{U}\times 1}$
    \While{$i \leq \tau_2$ and $j \leq \tilde{U}$}
    \State $allo = \min\{\textbf{sup}_i, \textbf{dem}_j \}$    
    \State ${\bm \Pi}_{i,j}=allo$
    \State $\textbf{sup}_i = \textbf{sup}_i - allo$
    \State $\textbf{dem}_i = \textbf{dem}_i - allo$
    \If{$\textbf{sup}_i$ = 0}
        \State $i=i+1$
    \ElsIf{$\textbf{dem}_i$ = 0}
        \State $j=j+1$
    \EndIf
    \EndWhile
\end{algorithmic}
\end{algorithm}

Once $\bm \Pi$ is constructed, it guides the design of the pilot signals $\bar{\bm a}_t$'s and the block matrices ${\bm \Xi}_t$'s for each time instant $t$. For the pilot signals, we set $\bar{a}_{j,\tau_1+i}=1$ if ${\bm \Pi}_{i,j}\neq 0$ and $\bar{a}_{j,\tau_1+i}=0$ otherwise, $i=1,\cdots,\tau_2$, $j=1,\cdots,\tilde{U}$, where ${\bm \Pi}_{i,j}$ denotes the element in the $i$-th row and $j$-th column of ${\bm \Pi}$. For the block matrices ${\bm \Xi}_t$'s, the basic idea is to select rows from an arbitrary unitary matrix ${\bm P}$ to form each ${\bm \Xi}_t$ in a structured manner dictated by the entries of $\bm \Pi$. The selection process follows these guidelines:
\begin{itemize}
    \item Directional Selection: Starting from ${\bm \Pi}_{1,1}$, we move in two possible directions: 1) Downward: A change in the row index corresponds to advancing to the next time instant $t$; 2) Rightward: A change in the column index indicates that additional rows are to be selected for the same time instant.
    \item Sequential Assignment: For each nonzero entry ${\bm \Pi}_{i,j}$, the value indicates the number of rows to be selected from ${\bm P}$ (sequentially from the first row onward, cycling back to the first row when all rows have been used). These selected rows are used to form ${\bm \Xi}_t$.
\end{itemize}
% The overall process of constructing the blocks ${\bm \Xi}_t$'s is summarized in Algorithm \ref{alg2}.
Along this way, we actually construct a series of sequence based on a primary sequence $S_{M,0}=\{1,2,\cdots,M\}$. For time instant $t$, a circular shifted version of $S_{M,0}$ is obtained as
\begin{align}
    S_{M,t} = {\rm circshift}(S_{M,0}, (t-\tau_1-1)q~{\rm mod}~M),
\end{align}
where $\rm mod$ is the modulo operation. Then, ${\bm \Phi}_t$ is constructed by
\begin{align}
    {\bm \Phi}_t = {\bm P}(S_{M,t},:),~~t=\tau_1+1,\cdots,\tau_1+\tau_2,
\end{align}
and ${\bm \Xi}_{t}={\bm \Sigma}{\bm V}^H{\bm \Phi}_t$. %{The above procedure for constructing $\tilde{\bm \Theta}_2$ is summarized in Algorithm \ref{alg2}.
To further explain the above procedure, we provide a simple example as follows.

\textit{Example 1:} Consider the case when $M=3$, $N=2$, $q=2$, $K=2$ and $U=2$. In this case, we have $\tilde{U}=3$ and  $\tau_2=5$. Based on Algorithm \ref{alg1}, the auxiliary matrix ${\bm \Pi}$ is built as:
\begin{align}\label{Pi_eg}
    {\bm \Pi}=\left[
    \begin{array}{ccc}
       2 & 0 & 0 \\
       1 & 1 & 0 \\
       0 & 2 & 0 \\
       0 & 0 & 2 \\
       0 & 0 & 1
    \end{array}
    \right]
\end{align}
% \begin{align}\label{Pi_eg}
%     {\bm \Pi}=\left[
%     \begin{array}{ccccc}
%         2 & & 0 & & 0 \\
%         \downarrow & & & & \\
%         1 & \rightarrow & 1 & & 0 \\
%         & & \downarrow & & \\
%         0 & & 2 & \rightarrow & 0 \\
%         & & & & \downarrow \\
%         0 & & 0 & & 2 \\
%         & & & & \downarrow \\
%         0 & & 0 & & 1
%     \end{array}
%     \right],
% \end{align}
We can check that the sum of each row in \eqref{Pi_eg} equals $q=2$ and the sum of each column equals $M=3$. Then, we select rows from ${\bm P}$ sequentially while cycling back when necessary. Specifically, in \eqref{Pi_eg}, starting from ${\bm \Pi}_{1,1}=2$, the first two rows of ${\bm P}$, i.e., ${\bm P}(1,:)$ and ${\bm P}(2,:)$ are selected and allocated to ${\bm \Xi}_{\tau_1+1}$. Next, for ${\bm \Pi}_{2,1}=1$, since the row index updates, the next available row, i.e., ${\bm P}(3,:)$, should be allocated to ${\bm \Xi}_{\tau_1+2}$. Then, for ${\bm \Pi}_{2,2}=1$, the sequence cycles back to the first row ${\bm P}(1,:)$. And since the column index updates, ${\bm P}(1,:)$ is also allocated to ${\bm \Xi}_{\tau_1+2}$, which is thus formed by concatenating ${\bm P}(3,:)$ and ${\bm P}(1,:)$. The process continues in this manner, and after the last non-zero entry of ${\bm \Pi}$ is considered, the resulted matrix $\tilde{\bm \Theta}_2$ is as follows:
\begin{align}\label{Theta_2_eg}
    &\tilde{\bm \Theta}_2=\notag\\
    &\left[
    \begin{array}{ccc}
        \bar{\bm \Sigma}{\bm P}^{'}([1,2],:) & {\bm O} & {\bm O}  \\
        \bar{\bm \Sigma}{\bm P}^{'}([3,1],:) & \bar{\bm \Sigma}{\bm P}^{'}([3,1],:) & {\bm O}  \\
        {\bm O} & \bar{\bm \Sigma}{\bm P}^{'}([2,3],:) & {\bm O} \\
        {\bm O} & {\bm O} & \bar{\bm \Sigma}{\bm P}^{'}([1,2],:) \\
        {\bm O} & {\bm O} & \sigma_1{\bm P}^{'}(3,:)        
    \end{array}
    \right],
\end{align}
where $\bar{\bm \Sigma} = {\rm diag}(\sigma_1,\sigma_2)$. It can be observed that: 1) Each of the $\tau_2$ row-blocks has a rank of $q=2$ and all the row-blocks are linearly independent; 2) Each of the $\tilde{U}$ column-blocks has a rank of $M=3$ and all the column-blocks are linearly independent. As a result, $\tilde{\bm \Theta}_2$ has a rank of $M\tilde{U}=9$. 

At last, once $\tilde{\bm \Theta}_2$ is constructed, ${\bm \Theta}_2$ is formed as ${\bm \Theta}_2=({\bm I}_{\tau_2}\otimes{\bm U})\tilde{\bm \Theta}_2$. Therefore, As long as $\tau_2\geq\lceil\frac{M(KU-1)}{q}\rceil$, we can find a proper design of $\bar{\bm a}_t$'s and ${\bm \Phi}_t$'s such that ${\rm rank}({\bm \Theta}_2)={\rm rank}(\tilde{\bm \Theta}_2)=M(KU-1)$. 

% \begin{algorithm}[t]
% \caption{Construction of $\tilde{\bm \Theta}_2$ via Block Matrices}\label{alg2}
% \begin{algorithmic}[1]
%         \Statex \textbf{Input:} auxiliary matrix ${\bm \Pi}$, unitary matrix ${\bm P}$
%         \Statex \textbf{Output:} $\tilde{\bm \Theta}_2$
%         \State $\bar{\bm a}_t \gets {\bm 0}_{\tilde{U} \times 1}$, ${\bm\Xi}_t \gets [~]$, $\forall t$
%         \State $i \gets 1$, $j \gets 1$, $n_{rows} \gets 0$, $n_q \gets 0$
%         \State ${\bm P}^{'} = {\bm V}^H{\bm P}$
%         \While{$i \leq \tau_2$ and $j \leq \tilde{U}$}
%         \State ${\bm \Xi}_{\tau_1+i} = [{\bm \Xi}_{\tau_1+i};{\rm diag}(\sigma_{{n_{q}}+1},\cdots,\sigma_{n_{q}+{\bm \Pi}_{i,j}})\cdot$ ${\bm P}^{'}(n_{rows}+1:n_{rows}+{\bm \Pi}_{i,j},:)]$
%         \State $n_{rows}=n_{rows}+{\bm\Pi}_{i,j}$
%         \State $n_q = n_q + {\bm\Pi}_{i,j}$
%         \State $\bar{a}_{j,\tau_1+i}=1$
%         \If{$n_q=q$}
%             \State $i=i+1$
%             \State $n_q = 0$
%         \EndIf
%         \If{$n_{rows} = M$}
%             \State $j=j+1$
%             \State $n_{rows}=0$
%         \EndIf
%         \EndWhile
%         \State $\tilde{\bm \Theta}_2$ is constructed based on \eqref{Theta_2_Xi}
%     \end{algorithmic}
% \end{algorithm}

\textit{Remark}: In our construction, the total number of supplies available is $\tau_2q$ and the total number of demands required is $M(KU-1)$. The problem of allocating $q$ rows to each time instant (supply constraint) and $M$ columns to each antenna (demand constraint) in $\tilde{\bm \Theta}_2$ is equivalent to a transportation problem. Since the overall supply is at least as large as the overall demand, i.e., $\tau_2q\geq M(KU-1)$, the North West Corner Rule will thus always provide a feasible allocation via generating an auxiliary matrix ${\bm \Pi}$ that guides the assignment of rows from the unitary matrix ${\bm P}$ to construct $\tilde{\bm \Theta}_2$.

\bibliographystyle{IEEEtran}
\bibliography{reference}

\begin{thebibliography}{10}
\providecommand{\url}[1]{#1}
\csname url@samestyle\endcsname
\providecommand{\newblock}{\relax}
\providecommand{\bibinfo}[2]{#2}
\providecommand{\BIBentrySTDinterwordspacing}{\spaceskip=0pt\relax}
\providecommand{\BIBentryALTinterwordstretchfactor}{4}
\providecommand{\BIBentryALTinterwordspacing}{\spaceskip=\fontdimen2\font plus
\BIBentryALTinterwordstretchfactor\fontdimen3\font minus
  \fontdimen4\font\relax}
\providecommand{\BIBforeignlanguage}[2]{{%
\expandafter\ifx\csname l@#1\endcsname\relax
\typeout{** WARNING: IEEEtran.bst: No hyphenation pattern has been}%
\typeout{** loaded for the language `#1'. Using the pattern for}%
\typeout{** the default language instead.}%
\else
\language=\csname l@#1\endcsname
\fi
#2}}
\providecommand{\BIBdecl}{\relax}
\BIBdecl

\bibitem{wcsp}
R.~Wang, S.~Zhang, and L.~Liang, ``Low-overhead channel estimation for beyond
  diagonal reconfigurable intelligent surface aided single-user
  communication,'' in \emph{Proc. IEEE Int. Conf. Wireless Commun. Signal
  Process. (WCSP)}, Oct. 2024.

\bibitem{Basar2019wireless}
E.~Basar, M.~Di~Renzo, J.~De~Rosny, M.~Debbah, M.-S. Alouini, and R.~Zhang,
  ``Wireless communications through reconfigurable intelligent surfaces,''
  \emph{IEEE Access}, vol.~7, pp. 116 753--116 773, Aug. 2019.

\bibitem{jian2022inte}
M.~Jian, G.~C. Alexandropoulos, E.~Basar, C.~Huang, R.~Liu, Y.~Liu, and
  C.~Yuen, ``Reconfigurable intelligent surfaces for wireless communications:
  Overview of hardware designs, channel models, and estimation techniques,''
  \emph{Intell. Conv. Networks}, vol.~3, no.~1, Mar 2022.

\bibitem{Li2023BDRIS}
H.~Li, S.~Shen, M.~Nerini, and B.~Clerckx, ``Reconfigurable intelligent
  surfaces 2.0: Beyond diagonal phase shift matrices,'' \emph{IEEE Commun.
  Mag.}, vol.~62, no.~3, pp. 102--108, Mar. 2024.

\bibitem{shen2022Modeling}
S.~Shen, B.~Clerckx, and R.~Murch, ``Modeling and architecture design of
  reconfigurable intelligent surfaces using scattering parameter network
  analysis,'' \emph{IEEE Trans. Wireless Commun.}, vol.~21, no.~2, pp.
  1229--1243, Feb. 2022.

\bibitem{li2023beyond}
H.~Li, S.~Shen, and B.~Clerckx, ``Beyond diagonal reconfigurable intelligent
  surfaces: From transmitting and reflecting modes to single-, group-, and
  fully-connected architectures,'' \emph{IEEE Trans. Wireless Commun.},
  vol.~22, no.~4, Apr. 2023.

\bibitem{zheyu2025}
\BIBentryALTinterwordspacing
Z.~Wu and B.~Clerckx, ``{Beyond Diagonal RIS in Multiuser MIMO: Graph Theoretic
  Modeling and Optimal Architectures with Low Complexity},'' Feb. 2025.
  [Online]. Available: \url{https://arxiv.org/pdf/2502.16509}
\BIBentrySTDinterwordspacing

\bibitem{Nerini2024universal}
M.~Nerini, S.~Shen, H.~Li, M.~Di~Renzo, and B.~Clerckx, ``A universal framework
  for multiport network analysis of reconfigurable intelligent surfaces,''
  \emph{IEEE Trans. Wireless Commun.}, vol.~23, no.~10, pp. 14\,575--14\,590,
  Jun. 2024.

\bibitem{Nerini2024graph}
M.~Nerini, S.~Shen, H.~Li, and B.~Clerckx, ``Beyond diagonal reconfigurable
  intelligent surfaces utilizing graph theory: Modeling, architecture design,
  and optimization,'' \emph{IEEE Trans. Wireless Commun.}, vol.~23, no.~8, pp.
  9972--9985, Feb. 2024.

\bibitem{Li2023mode}
H.~Li, S.~Shen, and B.~Clerckx, ``Beyond diagonal reconfigurable intelligent
  surfaces: A multi-sector mode enabling highly directional full-space wireless
  coverage,'' \emph{IEEE J. Sel. Areas Commnun.}, vol.~41, no.~8, pp.
  2446--2460, Aug. 2023.

\bibitem{Nerini2024closed}
M.~Nerini, S.~Shen, and B.~Clerckx, ``Closed-form global optimization of beyond
  diagonal reconfigurable intelligent surfaces,'' \emph{IEEE Trans. Wireless
  Commun.}, vol.~23, no.~2, pp. 1037--1051, Feb. 2024.

\bibitem{Santamaria2024MIMO}
I.~Santamaria, M.~Soleymani, E.~Jorswieck, and J.~Gutiérrez, ``{MIMO capacity
  maximization with beyond-diagonal RIS},'' in \emph{Proc. IEEE Int. Workshop
  Signal Process. Adv. Wireless Commun. (SPAWC)}, Oct. 2024.

\bibitem{samy2024diagonal}
\BIBentryALTinterwordspacing
M.~Samy, H.~Al-Hraishawi, A.~B.~M. Adam, K.~Ntontin, S.~Chatzinotas, and
  B.~Otteresten, ``{beyond diagonal RIS-aided networks: performance analysis
  and sectorization tradeoff},'' Jun. 2024. [Online]. Available:
  \url{https://arxiv.org/abs/2406.04009}
\BIBentrySTDinterwordspacing

\bibitem{björnson2025capacity}
\BIBentryALTinterwordspacing
E.~Björnson and Özlem Tuğfe~Demir, ``{Capacity maximization for MIMO
  channels assisted by beyond-diagonal RIS},'' Jan. 2025. [Online]. Available:
  \url{https://arxiv.org/abs/2411.18298}
\BIBentrySTDinterwordspacing

\bibitem{Nerini2023discrete}
M.~Nerini, S.~Shen, and B.~Clerckx, ``Discrete-value group and fully connected
  architectures for beyond diagonal reconfigurable intelligent surfaces,''
  \emph{IEEE Trans. Veh. Technol.}, vol.~72, no.~12, pp. 16\,354--16\,368, Dec.
  2023.

\bibitem{wang2024radar}
B.~Wang, H.~Li, S.~Shen, Z.~Cheng, and B.~Clerckx, ``A dual-function
  radar-communication system empowered by beyond diagonal reconfigurable
  intelligent surface,'' \emph{IEEE Trans. Commun.}, Aug. 2024.

\bibitem{raeisi2024localizaiton}
\BIBentryALTinterwordspacing
M.~Raeisi, H.~Chen, H.~Wymeersch, and E.~Basar, ``Efficient localization with
  base station-integrated beyond diagonal {RIS},'' Nov. 2024. [Online].
  Available: \url{https://arxiv.org/abs/2411.13295}
\BIBentrySTDinterwordspacing

\bibitem{chen2024transmitter}
K.~Chen and Y.~Mao, ``Transmitter side beyond-diagonal {RIS} for mmwave
  integrated sensing and communications,'' in \emph{Proc. IEEE Int. Workshop
  Signal Process. Adv. Wireless Commun. (SPAWC)}, Sep. 2024, pp. 951--955.

\bibitem{guang2024power}
Z.~Guang, Y.~Liu, Q.~Wu, W.~Wang, and Q.~Shi, ``Power minimization for {ISAC}
  system using beyond diagonal reconfigurable intelligent surface,'' \emph{IEEE
  Trans. Veh. Technol.}, vol.~73, no.~9, pp. 13\,950--13\,955, Apr. 2024.

\bibitem{wang2024attacks}
H.~Wang, Z.~Han, and A.~L. Swindlehurst, ``Channel reciprocity attacks using
  intelligent surfaces with non-diagonal phase shifts,'' \emph{IEEE Open J.
  Commun. Soc.}, vol.~5, pp. 1469--1485, Feb. 2024.

\bibitem{li2024fullduplex}
\BIBentryALTinterwordspacing
H.~Li and B.~Clerckx, ``{Non-reciprocal beyond diagonal RIS: Multiport network
  models and performance benefits in full-duplex systems},'' Nov. 2024.
  [Online]. Available: \url{https://arxiv.org/abs/2411.04370}
\BIBentrySTDinterwordspacing

\bibitem{wang2020channel}
Z.~Wang, L.~Liu, and S.~Cui, ``Channel estimation for intelligent reflecting
  surface assisted multiuser communications: Framework, algorithms, and
  analysis,'' \emph{IEEE Trans. Wireless Commun.}, vol.~19, no.~10, pp.
  6607--6620, Oct. 2020.

\bibitem{li2024channel}
H.~Li, S.~Shen, Y.~Zhang, and B.~Clerckx, ``Channel estimation and beamforming
  for beyond diagonal reconfigurable intelligent surfaces,'' \emph{IEEE Trans.
  Signal Process.}, vol.~72, pp. 3318--3332, Jul. 2024.

\bibitem{li2023channel}
H.~Li, Y.~Zhang, and B.~Clerckx, ``Channel estimation for beyond diagonal
  reconfigurable intelligent surfaces with group-connected architectures,'' in
  \emph{Proc. IEEE Int. Workshop Comput. Adv. Multi-Sens. Adapt. Process.
  (CAMSAP)}, Dec. 2023.

\bibitem{de2024channel}
\BIBentryALTinterwordspacing
A.~L.~F. de~Almeida, B.~Sokal, H.~Li, and B.~Clerckx, ``{Channel estimation for
  beyond diagonal RIS via tensor decomposition},'' Jul. 2024. [Online].
  Available: \url{https://arxiv.org/abs/2407.20402}
\BIBentrySTDinterwordspacing

\bibitem{dearaujo2024semiblind}
\BIBentryALTinterwordspacing
G.~T. de~Araujo and A.~L.~F. de~Almeida, ``Semi-blind channel estimation for
  beyond diagonal {RIS},'' Dec. 2024. [Online]. Available:
  \url{https://arxiv.org/abs/2412.02824}
\BIBentrySTDinterwordspacing

\bibitem{sokal2024decoupled}
\BIBentryALTinterwordspacing
B.~Sokal, Fazal-E-Asim, A.~L.~F. de~Almeida, H.~Li, and B.~Clerckx, ``A
  decoupled channel estimation method for beyond diagonal {RIS},'' Dec. 2024.
  [Online]. Available: \url{https://arxiv.org/abs/2412.06683}
\BIBentrySTDinterwordspacing

\bibitem{ginige2024prediction}
\BIBentryALTinterwordspacing
N.~Ginige, A.~S. de~Sena, N.~H. Mahmood, N.~Rajatheva, and M.~Latva-aho,
  ``{Efficient channel prediction for beyond diagonal RIS-assisted MIMO systems
  with channel aging},'' 2024. [Online]. Available:
  \url{https://arxiv.org/abs/2411.17725}
\BIBentrySTDinterwordspacing

\bibitem{henderson1981vec}
H.~V. Henderson and S.~R. Searle, ``The vec-permutation matrix, the vec
  operator and kronecker products: A review,'' \emph{Linear and Multilinear
  Algebra}, vol.~9, no.~4, pp. 271--288, 1981.

\bibitem{mishra2019channel}
D.~Mishra and H.~Johansson, ``Channel estimation and low-complexity beamforming
  design for passive intelligent surface assisted {MISO} wireless energy
  transfer,'' in \emph{Proc. IEEE Int. Conf. Acoust. Speech Signal Process.
  (ICASSP)}, May 2019.

\bibitem{kolda2009tensor}
T.~G. Kolda and B.~W. Bader, ``Tensor decompositions and applications,''
  \emph{SIAM review}, vol.~51, no.~3, pp. 455--500, 2009.

\bibitem{cichocki2009nonnegative}
A.~Cichocki, R.~Zdunek, A.~H. Phan, and S.-i. Amari, \emph{Nonnegative matrix
  and tensor factorizations: Applications to exploratory multi-way data
  analysis and blind source separation}.\hskip 1em plus 0.5em minus 0.4em\relax
  John Wiley \& Sons, 2009.

\end{thebibliography}

\end{document}